\documentclass[journal]{IEEEtran}

\IEEEoverridecommandlockouts

\usepackage[utf8]{inputenc} 

\usepackage[T1]{fontenc}
\usepackage{url}
\usepackage{ifthen}
\usepackage{cite}
\usepackage{overpic}
\usepackage[cmex10]{amsmath} 

\usepackage{textcase}
\usepackage[tablename=Table]{caption}
\usepackage{blindtext}
\usepackage{floatrow}
\usepackage{soul}
\usepackage{changes}
\usepackage{gensymb}
\usepackage{cite}
\usepackage{amsmath,amssymb,amsfonts}

\usepackage{textcomp}
\usepackage{graphicx}

\usepackage{amssymb}
\usepackage{amsfonts}
\usepackage{amsmath}
\usepackage{epsfig}
\usepackage{epigraph}
\usepackage{color}
\usepackage{fancybox}
\usepackage{textcomp}
\usepackage{multirow}
\usepackage{setspace}
\usepackage{psfrag}
\usepackage{booktabs}
\usepackage{float}
\usepackage[caption = false]{subfig}
\usepackage{algorithm}
\usepackage{algpseudocode}
\usepackage{mathtools, nccmath, bigints, amsfonts}
\usepackage{array}
\usepackage{stfloats}
\newcolumntype{L}{>{\centering\arraybackslash}m{3cm}}

\usepackage{amsthm}
\theoremstyle{definition}
\usepackage{xcolor}
\usepackage{mathrsfs}
\usepackage{tikz}
\usepackage{balance}

\newcommand{\Ttran}{^{\mathsf{T}}}

\newfloatcommand{capbtabbox}{table}[][0.4\textwidth]
 \usepackage[font={small},figurename=Fig.]{caption}

\newtheorem{corollary}{Corollary}
\newtheorem{Lemma}{Lemma}

\newtheorem{proposition}{Proposition}
\makeatletter 
\pretocmd\@bibitem{\color{black}\csname keycolor#1\endcsname}{}{\fail}
\newcommand\citecolor[1]{\@namedef{keycolor#1}{\color{black}}}
\makeatother

    \def\Complex{{\rm\rule[.23ex]{.03em}{1.1ex}\kern-.3em{C}}}

    \newcommand{\be}{\begin{equation}} \newcommand{\ee}{\end{equation}}
    \newcommand{\bea}{\begin{eqnarray}} \newcommand{\eea}{\end{eqnarray}}
    \newcommand{\benum}{\begin{enumerate}} \newcommand{\eenum}{\end{enumerate}}

        \newcommand{\diag}{{\sf diag}}
        
        \newcommand{\tr}{{\sf tr}}

\def\Htran{\mbox{\tiny $\mathrm{H}$}}
\def\Ttran{\mbox{\tiny $\mathrm{T}$}}
\def\CN{\mathcal{N}_{\mathbb{C}}} 
\def\imagunit{\mathsf{j}} 

\def\Ncl{N_{\textrm{path}}}
\def\imaginary{\mathsf{j}} 
\def\Htran{\mbox{\tiny $\mathrm{H}$}}
\def\Ttran{\mbox{\tiny $\mathrm{T}$}}
\def\EVM{\mbox{\small $\mathrm{EVM}$}}
\newcommand{\vect}[1]{\mathbf{#1}}
\begin{document}

\title{Performance Evaluation of Movable Antenna Arrays in Wideband Multi-User MIMO Systems }
\author{Amna~Irshad,~Emil~Bj{\"o}rnson,~\IEEEmembership{Fellow,~IEEE},~Alva~Kosasih,~Vitaly~Petrov \\
\thanks{A preliminary version has been presented at the IEEE~SPAWC~2025~\cite{IEEESPAWC2025}. A.~Irshad, E.~Bj{\"o}rnson, and V.~Petrov are with the Department of Communication Systems,  KTH Royal Institute of Technology, Stockholm, Sweden. A.~Kosasih is with the Nokia Technology Standards, Espoo, Finland. Email:~\{amnai,emilbjo,vitalyp\}@kth.se; alva.kosasih@nokia.com. The research was supported by the Grant 2022-04222 from the Swedish Research Council and by the SweWIN center (Vinnova grant 2023-00572).}

}

\maketitle

\begin{abstract}
Future wireless networks are expected to support increasingly high data rates and user densities, motivating advanced multi-antenna architectures capable of adapting to dynamic propagation environments. Movable antenna (MA) arrays have recently emerged as an extension of massive MIMO, enabling physical repositioning of antenna elements to better exploit spatial diversity and mitigate inter-user interference. While prior studies report promising gains under idealized assumptions, their performance under realistic wideband multi-user operation remains insufficiently understood. This paper presents a comprehensive evaluation of MA-enabled systems in practical uplink and downlink scenarios. A wideband OFDM system model is developed, and novel closed-form sum-rate expressions are derived for both uplink and downlink under linear and nonlinear processing. Hardware impairments are incorporated via an EVM-based model, from which a distortion-aware UL/DL duality is established and the resulting high-SNR sum rate ceiling is analytically characterized. In addition, the interactions between antenna position optimization, receiver processing, and user loading are examined, and performance is evaluated under both time-division duplexing (TDD) and frequency-division duplexing (FDD). The results show that movable antennas can provide noticeable gains in low-impairment regimes with strong multi-user interference, but these benefits are highly scenario-dependent and diminish under hardware-impairment-limited conditions or in rich-scattering environments. These findings highlight the importance of carefully assessing deployment conditions when considering antenna mobility as an alternative to conventional fixed array configurations.

\end{abstract}


\section{Introduction}
Future wireless networks are expected to employ larger antenna arrays at base stations (BSs), compared to the currently deployed massive multiple-input multiple-output (MIMO) systems, which exploit spatial multiplexing to increase data rates~\cite{Marzetta2010a,Bjornson2019c}. Favorable propagation in massive MIMO systems is often observed when $M \gg K$, where $M$ is the number of antennas and $K$ is the number of multiplexed user devices. With many more antennas than users, the array offers surplus spatial degrees of freedom, which tends to make user channels approximately orthogonal under rich scattering or sufficiently distinct spatial signatures. However, in practice, scaling to very large antenna counts at the BS is challenging due to rapidly increasing hardware costs and circuit power (e.g., many RF chains are needed), as well as higher baseband processing and interconnect demands~\cite{Rusek2013}.

One way to improve user separability without increasing the antenna count is to enlarge the effective aperture via sparse or distributed antenna placement. A larger aperture can improve angular resolution and reduce spatial correlation, albeit with potential drawbacks such as elevated sidelobes and, for (quasi-)uniform sparse layouts, grating lobes~\cite{GiDuck2007,Bjornson2019c}. For example, sparse uniform planar array configurations can realize a larger effective aperture using the same number of antenna elements. Other array geometries, such as non-uniform arrays, can also be designed to better exploit spatial degrees of freedom by tailoring the array response to optimize system metrics (e.g., data rate)~\cite{Amitay1968,Zhou2013,NULA2024}. Nevertheless, while signal processing (e.g., precoding/combining) can adapt to time-varying channels and user scheduling decisions, these fixed physical array geometries cannot change in response to dynamically varying propagation conditions and environments.

In contrast, movable antennas (MAs) have recently emerged as a promising extension of multi-antenna systems that can dynamically adapt the array geometry (and thus spatial selectivity) to the current user population and their propagation environment~\cite{Bjornson2019c,Ma2024,Ding2025}. By enabling the physical repositioning of antenna elements, MA-enabled arrays can better exploit spatial diversity and mitigate inter-user interference. Related concepts include fluid antenna systems~\cite{New2026}, which realize spatial reconfigurability via port/position selection and can be implemented using liquid-based structures or reconfigurable pixel antennas, as well as reconfigurable metasurfaces~\cite{irs1,irs2} that electronically tune their electromagnetic response. While differing in implementation, these architectures share the same fundamental principle as MA systems in introducing additional degrees of freedom through reconfigurable antenna placement, realized either through mechanical positioning (e.g., motors~\cite{movgeneral}) or reconfigurable pixel antennas~\cite{REMAA25}. The additional degree of freedom introduced by antenna mobility allows such systems to achieve performance gains beyond those of fixed uniform arrays, even with fewer antenna elements~\cite{xiao2023multiuser}. Hence, MAs have the potential to enable future MIMO systems that can serve more users without needing to increase the number of antennas~\cite{bjornson2026b}.

However, this flexibility of MAs comes at a non-trivial cost: practical MA implementation requires dedicated mechanical positioning hardware (e.g., motors or actuators), precise calibration, and accurate channel state information to realize meaningful performance gains, all of which add hardware complexity, power consumption, and latency overhead relative to fixed arrays. These practical challenges suggest that \emph{MAs are not universally superior}: e.g., in scenarios where the channel offers limited temporal and spatial variability, such as sparse or line-of-sight dominated environments with few resolvable paths and limited user mobility, the additional degrees of freedom introduced by antenna mobility may not translate into substantial gains over a well-designed fixed array, making the added complexity hard to justify. \emph{Consequently, understanding when and by how much MAs outperform fixed arrays across different channel conditions, array geometries, and system configurations becomes essential to assess their true potential.}

Several recent works have investigated movable antennas in different communication scenarios. The channel capacity of point-to-point MIMO systems aided by MAs was characterized in~\cite{gao2023mimo}, among other related works. The authors in~\cite{gao2023mimo} specifically considered a narrowband flat fading model and derived the achievable capacity under optimal antenna positioning, demonstrating notable capacity gains compared to a fixed antenna scheme. In~\cite{yao2024hard}, narrowband downlink MIMO systems with both fixed and movable antennas were further evaluated under practical hardware impairments. The study quantified the impact of non-ideal transceiver components on MA-enabled systems and showed that, despite hardware distortions, MAs can still provide performance improvements over conventional fixed-antenna deployments.

More recently, a few studies have explored wideband MA-enabled systems for single-antenna orthogonal frequency division multiplexing (OFDM) transmission. In~\cite{hong2025fluidantennaempowering5g}, a point-to-point single-antenna OFDM system with a fluid antenna was studied. The authors analyzed port selection strategies under frequency-selective fading and demonstrated that OFDM–fluid antenna systems achieve higher performance gains in relatively narrow bandwidth regimes, where highly dynamic channels limit the resolution and effectiveness of port selection. In~\cite{Zhu2024a}, a SISO system operating over frequency-selective fading channels was considered. The authors showed that MA-assisted systems achieve greater performance gains in scenarios with small delay spreads, typical of indoor environments. Subsequently,~\cite{Zhu2025a} extended this analysis by proposing a field-response-based channel model applicable to both narrowband and wideband systems. Their results indicated that, although MAs offer benefits in wideband systems, the relative gains are generally smaller than those observed in narrowband settings.

Despite many promising results reported to date, prior studies on MA systems primarily focus on specific system parameters: for instance, idealized narrowband settings, fixed receiver structures, line-of-sight (LoS) propagation environments~\cite{Shao20256dma}, and exclusively uplink~\cite{UplinkDing26,UplinkHu25} or downlink transmission~\cite{DownlinkChen26,cheng2023sumrate}. Moreover, most prior works do not provide a comprehensive analysis of practical wideband systems that incorporate hardware impairments, nonlinear processing, and different duplexing strategies (time division duplex (TDD) vs.\ frequency division duplex (FDD)). As a result, the robustness and general performance gains of MA-enabled systems under realistic wireless deployments remain insufficiently understood.

\subsection{Main Contributions}
Motivated by the above-mentioned gaps, this paper investigates MA-enabled systems in a comprehensive and practical multi-user wideband OFDM framework. We develop analytical models for both uplink and downlink transmissions under linear and nonlinear processing, explicitly accounting for hardware impairments and considering characteristics of both TDD and FDD operation.

The main contributions of this paper are as follows:
\begin{itemize}

    \item To the best of our knowledge, this is the first work to derive closed-form uplink and downlink sum-rate expressions for MA-assisted wideband multi-user MIMO systems with hardware impairments. We consider linear (i.e., minimum mean square error (MMSE)) and non-linear (i.e., successive interference cancelation (SIC) for uplink and dirty paper coding (DPC) for downlink) processing. The most closely related work~\cite{Zhu2025a} assumes ideal hardware and does not derive closed-form uplink and downlink sum-rate expressions under linear and non-linear multi-user processing. In contrast, our work explicitly incorporates hardware impairments, establishes distortion-aware UL/DL duality, and provides analytical sum rate characterizations for wideband multi-user MIMO systems under both TDD and FDD operation.

    \item With the use of the developed analytical evaluation framework, we determine the performance gains achieved by movable antennas in different deployment scenarios, hardware constraints, and algorithms implemented in the network. In particular, we characterize the impact of user load, channel conditions, and duplexing mode on the MA effectiveness. In particular, we compare the performance levels offered by MAs to those achievable with fixed planar array configurations, such as compact uniform planar array, sparse uniform planar array, and staggered uniform rectangular arrays. Whenever applicable, we also evaluate the gap between the performance levels achievable with MAs and the theoretical upper bounds (e.g., zero-interference scenario).
    
    \item Leveraging the above analytical and numerical results, we identify the operational regimes in which movable antennas provide substantial performance gains over fixed arrays. We particularly demonstrate that antenna mobility yields the largest gains in LoS-dominant scenarios with strong multi-user interference, while offering limited benefits in lightly loaded or distortion-limited regimes. These findings provide practical design guidelines for the deployment of MA systems.

\end{itemize}

The remainder of the paper is organized as follows. Section~\ref{sec:system_model} introduces the multi-user OFDM wideband system. Sections~\ref{sec:uplink} and \ref{sec:downlink} derive the closed form expressions for the uplink and downlink sumrate under linear and non-linear processing. Section~\ref{sec:problem-formulation} formulates the sum rate maximization problem and discusses the optimization of antenna locations for the MA system. Section~\ref{sec:results} provides the numerical performance evaluation in different scenarios and highlights the main findings. Finally, section~\ref{Sect_conclusion} concludes the paper.

\textit{Notation:} Scalars, vectors, and matrices are denoted by lowercase, bold lowercase, and bold uppercase letters, respectively. The Hermitian transpose is denoted by $(\cdot)^{\Htran}$ and expectation by $\mathbb{E}\{\cdot\}$. The transpose and imaginary unit are denoted by $(\cdot)^{\Ttran}$ and $\imaginary$, respectively. For a complex scalar $x$, $|x|$ denotes its modulus and $|\cdot|^2$ the squared Euclidean norm. The notation $a \triangleq b$ denotes equality by definition. 

\section{System Model}\label{sec:system_model}
\vspace{-0.5mm}

We consider a multi-user MIMO OFDM system with $K$ single-antenna users and a BS with an array of $M$ movable antennas, whose positions are $ \vect{p}_m \in \mathbb{R}^3$ for $m=1,\ldots,M$. The origin of the coordinate system is at the center of the array. We gather these position vectors in the matrix $\vect{P} = [ \vect{p}_1, \ldots,  \vect{p}_M] \in \mathbb{R}^{3 \times M}$, which will be optimized in this paper under constraints on the region where the antennas can be moved.
If a plane wave impinges on the BS from the azimuth angle-of-arrival (AOA) $\varphi$ and elevation AOA $\theta$, the array response vector is given as~\cite[Sec.~7.3.1]{massivemimobook}
\begin{equation} 
\vect{a}_{\vect{P}}(\varphi,\theta) = \begin{bmatrix} e^{\imaginary \, \vect{p}_1^{\Ttran}\vect{k}(\varphi,\theta) } & \ldots & e^{\imaginary \,\vect{p}_M^{\Ttran}\vect{k}(\varphi,\theta) }
\end{bmatrix}^{\Ttran},
\end{equation}
which depends on the position matrix $\vect{P}$ and the wave vector
\begin{equation} 
\vect{k}(\varphi,\theta) = \frac{2\pi}{\lambda} \begin{bmatrix}\cos(\varphi)  \cos(\theta)  \\  \sin(\varphi) \cos(\theta)\\ \sin(\theta) \end{bmatrix}.
\end{equation}

We use a geometric channel model to enable the optimization of $\vect{P}$. The channel from the $i$-th user has $\Ncl$ far-field paths, where the $n$th path is determined by: 1) the amplitude $\alpha_{i,n} \geq 0$; 2) the time delay $\tau_{i,n} \geq 0$; 3) the azimuth AOA $\varphi_{i,n} \in [-\pi,\pi]$; and 4) the elevation AOA $\theta_{i,n} \in [-\pi/2,\pi/2]$. Both the LoS path and the non-LoS (NLoS) paths are illustrated in Fig.~\ref{fig:uplink_multipath}.

\begin{figure*}[t!]
    \centering
    \subfloat[Uplink multipath channel setup\label{fig:uplink_multipath}]{
        \begin{overpic}[width=0.4\columnwidth,tics=10]{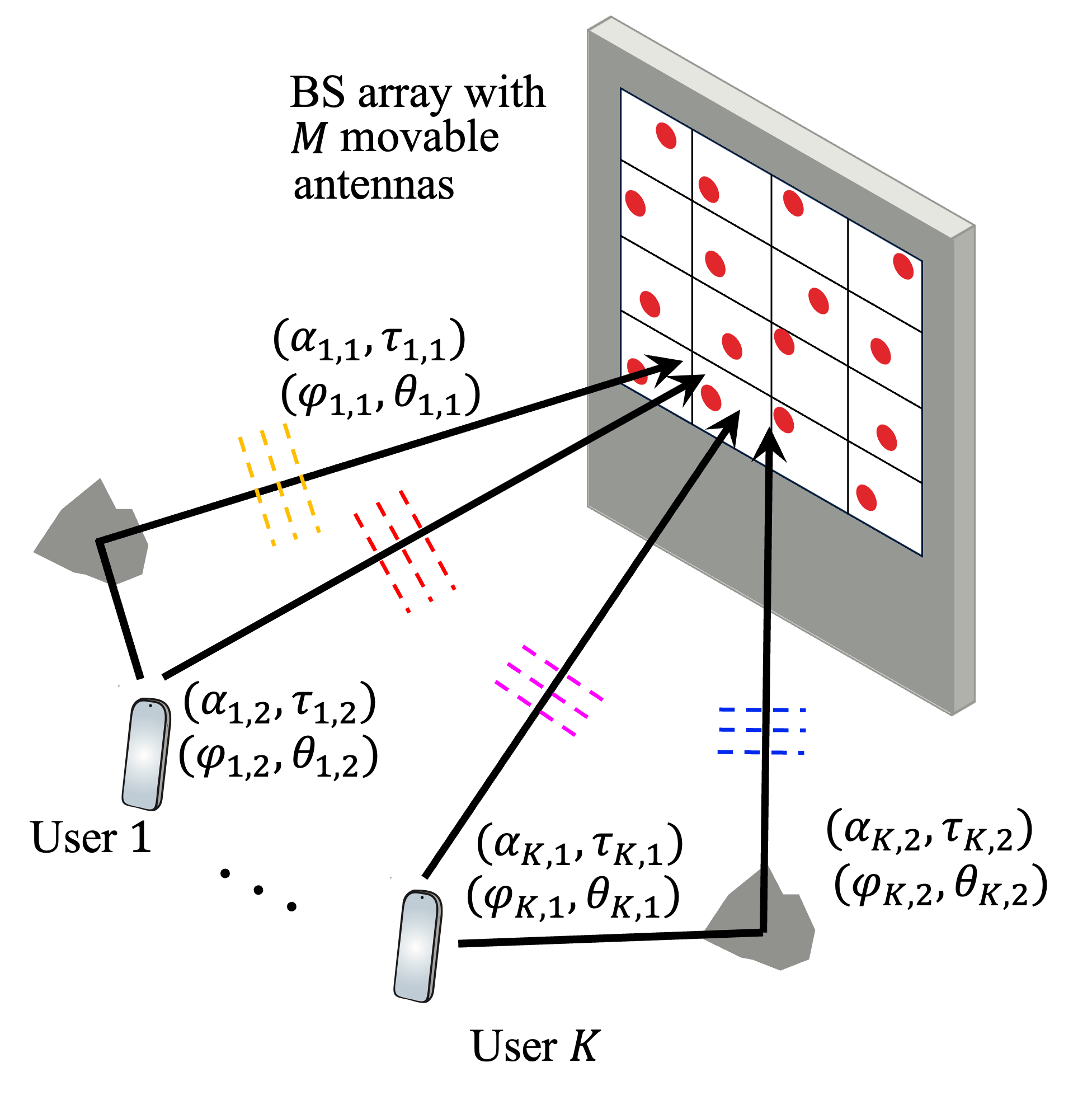}
        \end{overpic}}
    \hspace{1cm}
    \subfloat[Downlink multipath channel setup\label{fig:downlink_multipath}]{
        \begin{overpic}[width=0.4\columnwidth,tics=10]{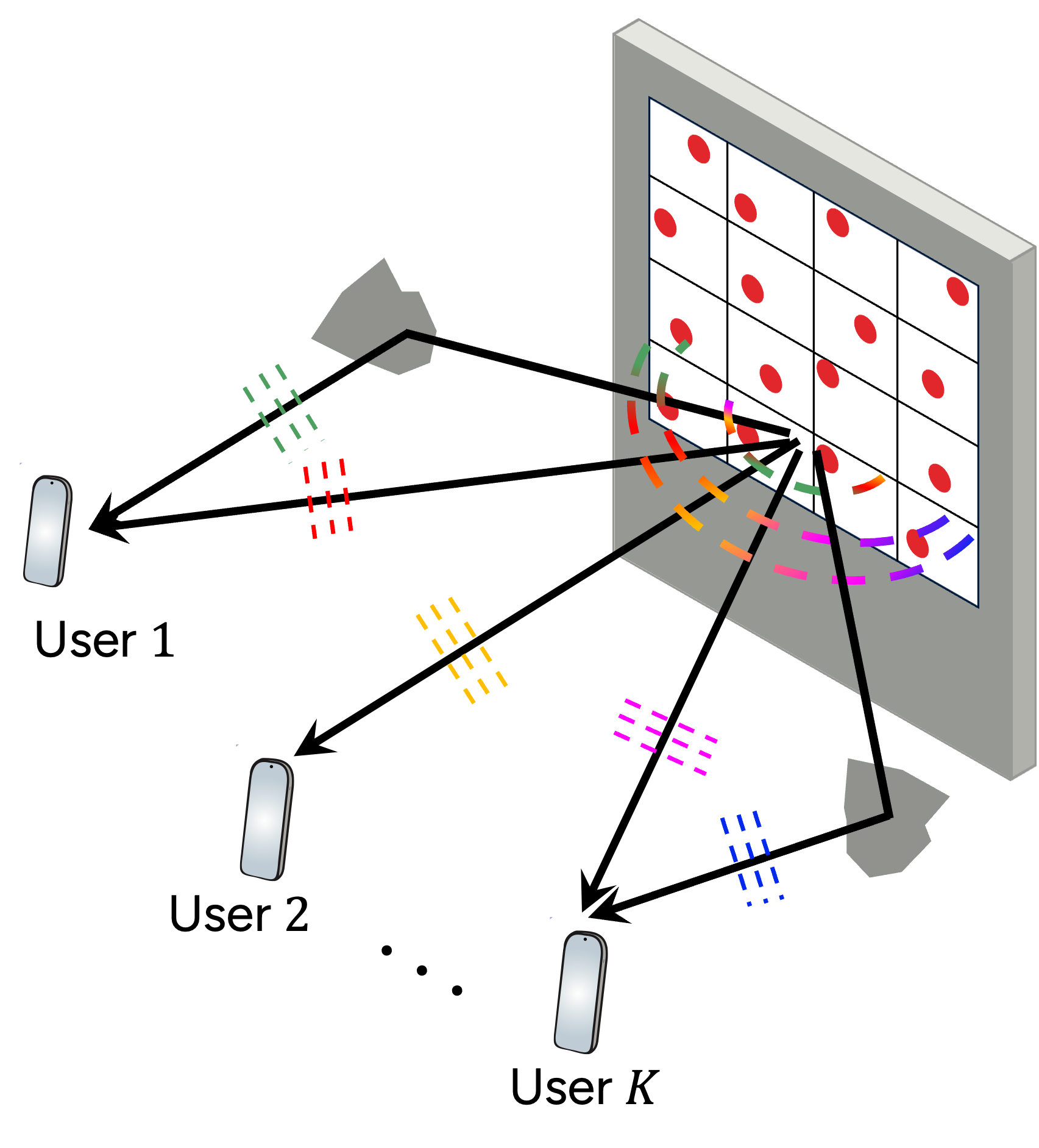}
        \end{overpic}}

    \caption{Illustration of the considered multi-suer multipath channel setups.}
\end{figure*}

The OFDM waveform uses $S$ subcarriers with a subcarrier spacing of $\Delta$ and a pulse-shaping filter $f(t)$ that is only non-zero for $t \in [-1,1]$.\footnote{One example is the triangle function $f(t) = 1-|t|$ for $t \in [-1,1]$ (and $0$ elsewhere), which is obtained by using box functions at the transmitter and receiver. Another practical example is a time-windowed raised-cosine filter.}
The channel from $i$-th user then becomes a finite impulse response filter with the taps~\cite[Ch.~7]{bookEmil}
\vspace{-0.5mm}
\begin{equation} \label{eq:hk-multitap}
\vect{h}_i[\ell] =  \sum_{n=1}^{\Ncl} b_{i,n}[\ell] \vect{a}_{\vect{P}}(\varphi_{i,n},\theta_{i,n}), \quad \ell=0,\ldots,T,
\end{equation}
where the scalar coefficients $b_{i,n}[\ell] \in \mathbb{C}$ are given by
\vspace{-0.5mm}
\begin{equation}
b_{i,n}[\ell] = \alpha_{i,n} e^{- \imagunit 2 \pi \lambda (\tau_{i,n} - \eta)/c} f \left(\ell+S\Delta(\eta-\tau_{i,n}) \right),
\end{equation}
the speed of light is denoted by $c$, the time synchronization coefficient at the receiver is $\eta = \min_{i,n} \tau_{i,n}$  (matched to the fastest path), the number of delay taps is $T = \left\lceil S\Delta (\max_{i,n} \tau_{i,n} - \eta) \right\rceil$, and $\lceil\cdot \rceil$ is the ceiling operation. A specific feature of this new OFDM channel model is that we can both vary the bandwidth $S\Delta$ (by changing the number of subcarriers $S$) and the antenna positions (by changing $\vect{P}$) while maintaining a spatially consistent propagation scenario.
The number of taps increases with the bandwidth, making the individual paths more distinguishable.

\section{Uplink Rates with Hardware Distortion}\label{sec:uplink}

In this section, we derive novel uplink rate expressions for OFDM systems with hardware impairments and movable antennas, accounting for different receiver combining schemes.
The received uplink signal $\bar{\vect{y}}[\nu] \in \mathbb{C}^M$ on subcarrier $\nu$ is
\vspace{-0.5mm}
\begin{equation} \label{eq:received-signal-MIMO-OFDM}
\bar{\vect{y}}[\nu] =  \sum_{i=1}^{K} \bar{\vect{h}}_i [\nu] \bar{\chi}_i[\nu] + \bar{\vect{n}}[\nu], \quad \nu =0,\ldots,S-1,
\end{equation}
where $\bar{\chi}_i[\nu]$ is the signal transmitted by the $i$-th user with power $\rho_i[\nu]$ on subcarrier $\nu$, $\bar{\vect{n}}[\nu] \sim \CN ( \vect{0}, \sigma^2 \vect{I}_M)$ is independent noise, and the channel from the $i$-th user obtain as
\begin{equation} \label{eq:H-frequency-domain-multipath}
\bar{\vect{h}}_i[\nu] = \sum_{n=1}^{\Ncl}  \left( \sum_{\ell = 0}^{T} b_{i,n}[\ell]  e^{-\imaginary 2 \pi \ell \nu /S}\right)
\vect{a}_{\vect{P}}(\varphi_{i,n},\theta_{i,n}), 
\end{equation}
by taking the discrete Fourier transform (DFT) of \eqref{eq:hk-multitap}.

Typical hardware components in wireless transceivers, such as power amplifiers, local oscillators, and analog-to-digital converters, introduce nonlinearities and imperfections. These give rise to residual distortions such as PA nonlinearities, phase noise, and finite-resolution quantization~\cite{HFimpairment,SchenkTim2008RIiH,ValkamaMikko2010RICf}, which degrade the communication performance. To quantify these effects in their totality, we will use the established error vector magnitude (EVM) metric~\cite[Sec.~6.1]{massivemimobook}.
We model the transmitted signal by $i$-th user on subcarrier $\nu$ as
\begin{align} \nonumber
\bar{\chi}_i[\nu] = &\sqrt{\rho_i[\nu]} \left(\sqrt{1-\EVM^2} d_i[\nu] + \epsilon_i[\nu]\right)\\ 
= &\sqrt{\rho_i[\nu]} \left(\sqrt{\kappa} d_i[\nu] + \epsilon_i[\nu]\right), \label{eq:uplink-signal}
\end{align}
where $d_i[\nu]$ is the distortionless unit-variance data signal, $\mathbb{E} \{ |d_i[\nu]|^2\} =1$, $\epsilon_i[\nu]$ is the uncorrelated additive distortion noise with variance $\EVM^2$, and $\mathbb{E} \{ |\epsilon_i[\nu]|^2\} =\EVM^2=
1-\kappa$, 
where $\kappa \triangleq 1-\EVM^2$ denotes the fraction of useful signal power. The total transmit power is $\mathbb{E} \{ |\bar{\chi}_i[\nu]|^2\} =  \rho_i[\nu] ( 1-\EVM^2 +\EVM^2) = \rho_i[\nu]$ regardless of the EVM, but the coefficient $\EVM \in [0,1]$ determines the fraction of distortion.

Substituting \eqref{eq:uplink-signal} into \eqref{eq:received-signal-MIMO-OFDM},
the received signal at the BS on subcarrier $\nu$ can be expressed as
\begin{align}
   \bar{\vect{y}}[\nu] =  \sum_{i=1}^K\sqrt{\rho_i[\nu]} \bar{\vect{h}}_i [\nu] \left(\sqrt{\kappa} d_i[\nu] + \epsilon_i[\nu]\right) + \bar{\vect{n}}[\nu].
\end{align}

In practice, there will also be impairments in the BS hardware, but it typically drowns in the distortions generated by the users~\cite{massivemimobook,Bjornson2019b}, because the multi-antenna processing spreads it out spatially and the BS usually has higher-grade hardware. To keep the analysis tractable in this paper, we will not include the BS hardware impairments in the system model.

\subsection{Linear receive combining}

We first consider the case where the BS uses linear receiver processing.
For the $k$-th user, the BS applies the receive combining vector $\vect{w}_k[\nu] \in \mathbb{C}^M$ on subcarrier $\nu$ such that
\begin{align}\nonumber
   \vect{w}_k[\nu]^{\Htran}\bar{\vect{y}}[\nu] =& \sqrt{\kappa\rho_k[\nu]} \vect{w}^{\Htran}_k[\nu]\bar{\vect{h}}_k[\nu]d_k[\nu]+\\ \nonumber&\sqrt{\kappa}\sum_{i \neq k} \sqrt{\rho_i[\nu]}\vect{w}^{\Htran}_k[\nu]\bar{\vect{h}}_i[\nu]d_i[\nu]\\ &+\sum_{i=1}^K \sqrt{\rho_i[\nu]}\vect{w}^{\Htran}_k[\nu]\bar{\vect{h}}_i[\nu]\epsilon_i[\nu]+\vect{w}^{\Htran}_k[\nu]\bar{\vect{n}}.
\end{align}
If the BS decodes the signal from the $k$-th user at subcarrier $\nu$ using linear receive combining and treats multiuser interference, distortion noise, and thermal noise as additive Gaussian noise, then an achievable rate is given by (see,~\cite[Ch. 2]{Tse_Viswanath_2005})
\begin{equation}
    R_k^{\mathrm{UL,lin}}= \log_2(1+\mathrm{SINR}_k^{\mathrm{UL,lin}}[\nu]).
\end{equation}
where $\mathrm{SINR}_k^{\mathrm{UL,lin}}[\nu]$ denotes the signal-to-interference-plus-noise-ratio (SINR). With an arbitrary combining vector $\vect{w}_k[\nu]$, the SINR equals \eqref{SINR_UL_LIN}, given at the top of next page.

\begin{figure*}

\begin{align}\nonumber
  \mathrm{SINR}_k^{\mathrm{UL,lin}}[\nu]=&\frac{\kappa\rho_k[\nu]|\vect{w}_{k}^{\Htran}[\nu]\bar{\vect{h}}_k[\nu]|^2}{\kappa \sum_{i\neq k}\rho_i[\nu]|\vect{w}_k^{\Htran}[\nu] \bar{\vect{h}}_i[\nu]|^2+(1-\kappa)\sum_{i= 1}^{K}\rho_i[\nu]|\vect{w}_k^{\Htran}[\nu]\bar{\vect{h}}_i[\nu]|^2+\sigma^2\|\vect{w}_k[\nu]\|^2}\\&\label{SINR_UL_LIN}
  = \frac{\kappa\rho_k[\nu]|\vect{w}_{k}^{\Htran}[\nu]\bar{\vect{h}}_k[\nu]|^2}{ \sum_{i\neq k}\rho_i[\nu]|\vect{w}_k^{\Htran}[\nu] \bar{\vect{h}}_i[\nu]|^2+(1-\kappa)\rho_k[\nu]|\vect{w}_k^{\Htran}[\nu]\bar{\vect{h}}_k[\nu]|^2+\sigma^2\|\vect{w}_k[\nu]\|^2}
\end{align}
\end{figure*}
Assuming perfect channel state information (to focus on the ultimate performance gains provided by MAs), and averaging the corresponding data rates over users and subcarriers, and considering the aggregate uncorrelated disturbance (multi-user interference, distortion and thermal noise) as Gaussian noise, the achievable sum rate is  
\begin{equation}\label{sum_rate_lin}
    R_\Sigma^{\mathrm{UL,lin}}=\frac{1}{S}\sum_{\nu=0}^{S-1}\sum_{k=1}^K \log_2(1+\mathrm{SINR}_k^{\mathrm{UL,lin}}[\nu]).
\end{equation}
The choice of the receive combining vector critically impacts the achieved rate.
Since each vector only affects one SINR, the optimal vector maximizes $\mathrm{SINR}_k^{\mathrm{UL,lin}}[\nu]$ in \eqref{SINR_UL_LIN} is the MMSE combiner~\cite{flex24}. Then, the SINR in \eqref{SINR_UL_LIN} can be written as
\begin{align}
  \mathrm{SINR}_k^{\mathrm{UL,lin}}[\nu]=&\frac{\vect{w}_{k}^{\Htran}[\nu](\kappa\rho_k[\nu]\bar{\vect{h}}_k[\nu]\bar{\vect{h}}_k^{\Htran}[\nu])\vect{w}_{k}[\nu]}{\vect{w}^{\Htran}_{k}[\nu]\vect{Q}_{i+n}^k[\nu]\vect{w}_{k}[\nu]},
\end{align}
where the disturbance covariance matrix is stated in \eqref{Q_I_N} on the next page.

\begin{figure*}
    \begin{align}\label{Q_I_N}
 \vect{Q}_{i+n}^k[\nu] &=\sum_{i\neq k}\rho_i[\nu]\bar{\vect{h}}_i[\nu]\bar{\vect{h}}^{\Htran}_i[\nu]+(1-\kappa)\rho_k[\nu]\bar{\vect{h}}_k[\nu]\bar{\vect{h}}^{\Htran}_k[\nu]+\sigma^2\vect{I}_M
\end{align}
\hrulefill
\end{figure*}
The SINR is therefore a generalized Rayleigh quotient with respect to $\vect{w}_k[\nu]$ which is maximized by 
\begin{align}
\vect{w}_k^{\mathrm{MMSE}}[\nu] = \vect{Q}_{i+n}^k[\nu]^{-1} \bar{\vect{h}}_k[\nu].
\end{align}

\subsection{Non-linear decoding with interference cancellation}
While linear receivers are attractive due to their low complexity, treating multiuser interference as noise is generally suboptimal. To characterize the ultimate uplink performance, we now consider optimal non-linear decoding via successive interference cancellation (SIC). SIC achieves the sum capacity of the multiple-access channel under a sum-power constraint. Hence, it serves as a theoretical optimal nonlinear benchmark, allowing us to characterize the fundamental performance limits of the uplink MA-assisted system under hardware distortion. We note that the received signal in \eqref{eq:received-signal-MIMO-OFDM} has the same form as a point-to-point MIMO system, where it is represented as
\begin{equation}
\bar{\vect{y}}[\nu] = \sqrt{\kappa}\bar{\vect{H}}_\rho[\nu]\vect{d}[\nu] + \underbrace{\bar{\vect{H}}_\rho[\nu]\vect{\epsilon}[\nu] + \bar{\vect{n}}[\nu]}_{\text{effective noise } \tilde{\vect{n}}[\nu]}.
\end{equation}
Here $\bar{\vect{H}}_\rho[\nu] = [\sqrt{\rho_1[\nu]}\bar{\vect{h}}_1[\nu], \sqrt{\rho_2[\nu]}\bar{\vect{h}}_2[\nu], \ldots, \sqrt{\rho_K[\nu]}\bar{\vect{h}}_K[\nu]]$, and the effective noise $\tilde{\vect{n}}[\nu]$ has the covariance matrix
\begin{align}
\vect{Q}[\nu] = \mathbb{E}\left[{\tilde{\vect{n}}[\nu]\tilde{\vect{n}}[\nu]^{\Htran}}\right] = (1-\kappa)\bar{\vect{H}}_\rho[\nu]\bar{\vect{H}}_\rho[\nu]^{\Htran} + \sigma^2\vect{I}_M. \label{eq:Q-expression}
\end{align}
With SIC, the uplink sum capacity equals the mutual information of the equivalent MIMO channel. Applying the worst-case uncorrelated additive noise theorem~\cite{Hassibi2003a}, the achievable sum rate for Gaussian signaling on subcarrier $\nu$  becomes
\begin{align} \nonumber
R [\nu] &= \log_2 \det \left( \vect{I}_M + \kappa \vect{Q}^{-1}[\nu] \bar{\vect{H}}_\rho[\nu] 
\bar{\vect{H}}_\rho^{\Htran}[\nu] \right) \\ \nonumber
&=\log_2 \det \left( \left[\vect{Q}[\nu] + \kappa  \bar{\vect{H}}_\rho[\nu] 
\bar{\vect{H}}_\rho^{\Htran}[\nu] \right]\vect{Q}^{-1}[\nu] \right) \\ \nonumber
&=\log_2 \det \left( \vect{Q}[\nu] + \kappa  \bar{\vect{H}}_\rho[\nu] 
\bar{\vect{H}}_\rho^{\Htran}[\nu] \right)  -\log_2 \det\left( \vect{Q}[\nu]\right) \\ \nonumber
&= \log_2 \det \left( \frac{ \kappa}{\sigma^2}  \bar{\vect{H}}_\rho[\nu] \bar{\vect{H}}_\rho^{\Htran}[\nu] + \frac{1}{\sigma^2} \vect{Q}[\nu] \right) \\
& \quad - \log_2\det \left( \frac{1}{\sigma^2} \vect{Q}[\nu] \right),
\end{align}
where we utilized the facts that $\det(\vect{AB})=\det(\vect{A})\det(\vect{B})$ and $\log_2(ab)=\log_2(a)+\log_2(b)$. 
Using \eqref{eq:Q-expression}, we can simplify the expression as
\begin{align} \nonumber
R [\nu] &= \log_2 \det \left( \frac{ \kappa + (1-\kappa)}{\sigma^2} \bar{\vect{H}}_\rho[\nu]\bar{\vect{H}}_\rho^{\Htran}[\nu] + \vect{I}_M \right) \\ \nonumber
&\quad - \log_2\det \left( \frac{(1-\kappa)}{\sigma^2} \bar{\vect{H}}_\rho[\nu]\bar{\vect{H}}_\rho^{\Htran}[\nu] + \vect{I}_M \right) \\ \nonumber
&= \log_2 \det \left( \vect{I}_M + \frac{1}{\sigma^2} \bar{\vect{H}}_\rho[\nu]\bar{\vect{H}}_\rho^{\Htran}[\nu] \right) \\
&\quad - \log_2 \det \left( \vect{I}_M + \frac{(1-\kappa)}{\sigma^2} \bar{\vect{H}}_\rho[\nu]\bar{\vect{H}}_\rho^{\Htran}[\nu] \right). \label{eq:uplink-sum-rate-SIC}
\end{align}
Note that $\bar{\vect{H}}_\rho[\nu]$ can be written as $\bar{\vect{H}}_\rho[\nu]=\bar{\vect{H}}[\nu]\vect{D}_\rho^{1/2}[\nu]$, where $\bar{\vect{H}}[\nu] = [\bar{\vect{h}}_1[\nu], \ldots, \bar{\vect{h}}_K[\nu]]$ is the full channel matrix and $\vect{D}_\rho[\nu]= \diag(\rho_1[\nu], \ldots, \rho_K[\nu])$ is the diagonal power allocation matrix, which gives $\bar{\vect{H}}_\rho[\nu]\bar{\vect{H}}_\rho^{\Htran}[\nu]=\bar{\vect{H}}[\nu]\vect{D}_\rho[\nu]\bar{\vect{H}}^{\Htran}[\nu]$. Averaging \eqref{eq:uplink-sum-rate-SIC} over the $S$ subcarriers, we obtain the average sum rate

\begin{align} \nonumber
R_{\Sigma}^{\mathrm{UL,SIC}} &=  \frac{1}{S} \sum_{\nu=0}^{S-1} \log_2 \det \left( \vect{I}_M + \frac{1}{\sigma^2} \bar{\vect{H}}[\nu] \vect{D}_\rho[\nu] \bar{\vect{H}}^{\Htran}[\nu] \right)-\\
&\quad \frac{1}{S} \sum_{\nu=0}^{S-1} \log_2 \det \left( \vect{I}_M + \frac{ (1-\kappa)}{\sigma^2} \bar{\vect{H}}[\nu]\vect{D}_\rho[\nu] \bar{\vect{H}}^{\Htran}[\nu] \right).\label{eq:sumrate-expression}
\end{align}

This new expression has two terms, where the first represents the sum rate with ideal hardware and the second is the penalty imposed by hardware impairments. 

\begin{proposition} \label{cor1}
Suppose the power allocation is selected as $\vect{D}_\rho[\nu] =\rho\vect{A}[\nu]$, where $\vect{A}[\nu]=\diag(a_1[\nu],\ldots,a_K[\nu])$ is a fixed matrix and $\rho$ is a parameter.
If $\bar{\vect{H}}[\nu]$ has rank $K$ and $\EVM>0$, then
\begin{align} \label{eq:upper-limit_evm}
\lim_{\rho \to \infty} R_{\Sigma} =  K \log_2 \left(\frac{1}{\EVM^2} \right).
\end{align}
\end{proposition}
\begin{IEEEproof}
 With the given power allocation, the term $\bar{\vect{H}}[\nu] \vect{D}_\rho[\nu] \bar{\vect{H}}^{\Htran}[\nu]$ becomes $\rho\bar{\vect{H}}[\nu] \vect{A}[\nu] \bar{\vect{H}}^{\Htran}[\nu]$. Introducing the notation $\vect{B}[\nu]=\bar{\vect{H}}[\nu] \vect{A}[\nu] \bar{\vect{H}}^{\Htran}[\nu]$, the sum rate becomes
\begin{align*}
    R_{\Sigma} &= \frac{1}{S} \sum_{\nu=0}^{S-1} \log_2 \det \left(  \vect{I}_M +  \frac{\rho}{\sigma^2}\vect{B}[\nu] \right) \\
& \quad -  \frac{1}{S}  \sum_{\nu=0}^{S-1}\log_2 \det \left(\vect{I}_M +  \frac{(1-\kappa)\rho}{\sigma^2 } \vect{B}[\nu] \right).
\end{align*}
Let $\lambda_1[\nu],\ldots,\lambda_K[\nu]$ denote the non-zero eigenvalues of $\vect{B}[\nu]=\vect{U}[\nu]\vect{\Lambda}[\nu]\vect{U}^{\Htran}[\nu]$, where $\vect{\Lambda}[\nu]=\diag(\lambda_1[\nu],\ldots,\lambda_K[\nu])$.
Utilizing determinant property,
\begin{align}
\det(\vect{I}+c\vect{U}[\nu]\vect{\Lambda}[\nu]\vect{U}^{\Htran}[\nu])=\det(\vect{I}+c\vect{\Lambda})=\prod_{k=1}^K (1+c\lambda_k[\nu]),
\end{align}
we obtain
\begin{align}
\notag
 R_{\Sigma} &= \frac{1}{S}  \sum_{\nu=0}^{S-1}\sum_{k=1}^K \log_2  \left(  1 +  \frac{\rho}{\sigma^2}\lambda_k[\nu] \right) \\ 
& \notag \quad -  \frac{1}{S}  \sum_{\nu=0}^{S-1}\sum_{k=1}^K\log_2 \left(1 +  \frac{(1-\kappa)\rho}{\sigma^2 } \lambda_k[\nu] \right) \\
&=  \frac{1}{S}  \sum_{\nu=0}^{S-1}\sum_{k=1}^K \log_2 \left(  
\frac{1 +  \frac{\rho}{\sigma^2}\lambda_k[\nu]}{1 +  \frac{(1-\kappa)\rho}{\sigma^2 } \lambda_k[\nu]}
 \right).
\label{R_sum}
\end{align}
We can now take the limit and obtain

\begin{align}
 \lim_{\rho \to \infty} R_{\Sigma} &= \frac{1}{S} \sum_{\nu=0}^{S-1} \sum_{k=1}^K
 \log_2 \!\left(  
\frac{\frac{\lambda_k[\nu]}{\sigma^2}}{\frac{(1-\kappa) \lambda_k[\nu]}{\sigma^2 } }
 \right)=K\log_2 \!\left(\frac{1}{1-\kappa}\right),
\end{align}
resulting in \eqref{eq:upper-limit_evm}. This completes the proof.
\end{IEEEproof}

This proposition demonstrates that there is an upper bound on the achievable rate when the transmit power grows large.
The upper bound depends on the EVM, but is independent of the BS antenna positions; thus, their positions only determines the performance at finite SNRs. 

\section{Downlink Rates with Hardware Distortion}\label{sec:downlink}

In the previous section, we derived the uplink rates with linear and non-linear processing. We now consider the corresponding downlink transmission where the BS employs either linear precoding or DPC. 

\subsection{Linear transmit precoding}
When the BS transmits to $K$ users using linear precoding, the transmitted signal at subcarrier $\nu$ can be denoted as
\begin{equation}
  \vect{x}[\nu]=\mathbf{P}[\nu]\vect{d[\nu]}=\sum_{i=1}^K \vect{p}_i[\nu] d_i[\nu],
\end{equation}
where $\vect{x}[\nu] \in \mathbb{C}^{M}$, $\mathbf{P}[\nu] \in \mathbb{C}^{M \times K}$ and $\vect{d}[\nu] \in \mathbb{C}^{K }$. Here, $\mathbf{P}[\nu]=[\vect{p}_1[\nu],\ldots,\vect{p}_i[\nu],\ldots,\vect{p}_K[\nu]]\in \mathbb{C}^{M \times K}$ is the precoding matrix, $\vect{p}_i[\nu]$ is the precoding vector for the $i$-th user at subcarrier $\nu$ and $\vect{d}[\nu]\in\mathbb{C}^K$ is the data vector with $\mathbb{E}\{|d_i[\nu]|^2\}=1$. The total BS transmit power is constrained as $\sum_{i=1}^K\|\vect{p}_i[\nu]\|^2\leq P_{\mathrm{tot}}$. 
Without hardware impairments, the received signal at the $k$-th user is
\begin{equation}
y_k^{\mathrm{ideal}}[\nu] = \sum_{i=1}^K \bar{\vect{h}}_k[\nu]^{\Htran}\vect{p}_i[\nu] d_i[\nu] + n_k[\nu],
\end{equation}
where $\bar{\vect{h}}_k[\nu]\in\mathbb{C}^M$ is the channel vector that was defined in the uplink and $n_k[\nu]\sim \mathcal{CN}(0,\sigma^2)$ is thermal noise. With hardware impairments at the user side (similar to the uplink), the effective received signal becomes~\cite[Ch.~6]{massivemimobook}
\begin{align}\nonumber
\bar{y}_k[\nu] =& \sqrt{\kappa}\, y_k^{\mathrm{ideal}}[\nu] + \eta_k[\nu]\\=& \nonumber
\sqrt{\kappa} \bar{\vect{h}}_k^{\Htran}[\nu]\vect{p}_k[\nu] d_k[\nu]+ \sqrt{\kappa}\sum_{i\neq k} \bar{\vect{h}}_k^{\Htran}[\nu]\vect{p}_i[\nu] d_i[\nu] \\\label{UEmodel}&+ \sqrt{\kappa} n_k[\nu]+ \eta_k[\nu]
,
\end{align}
where $\eta_k[\nu]\sim \mathcal{CN}(0,(1-\kappa)\mathbb{E}\{|y_k^{\mathrm{ideal}}[\nu]|^2\})$ is an additive distortion term that is uncorrelated with $y_k^{\mathrm{ideal}}[\nu]$, where $\mathbb{E}\{|y_k[\nu]|^2\})=\mathbb{E}\{|y_k^{\mathrm{ideal}}[\nu]|^2\})$.
When the distortion is proportional to the total received power
\begin{equation}
\mathbb{E}\{|\eta_k[\nu]|^2\} = (1-\kappa)\left(\sum_{i=1}^K |\bar{\vect{h}}_k[\nu]^{\Htran}\vect{p}_i[\nu]|^2 + \sigma^2\right).
\end{equation}

By treating the distortion and multi-user interference as worst-case uncorrelated Gaussian noise~\cite{Hassibi2003a}, the achievable rate at the $k$-th user is $\log_2(1+\mathrm{SINR}_k^{\mathrm{DL,lin}})$, where $\mathrm{SINR}_k^{\mathrm{DL,lin}}$ is given in \eqref{SINR_DL_lin}, provided at the top of the page. When $\kappa=1$ (i.e., no hardware impairments), this expression reduces to the classical SINR for downlink linear precoding~\cite{bjornson2017massive}. 

The distortion term scales with the total received power and therefore limits the achievable SINR at high SNR. The corresponding achievable sum rate is
\begin{equation}\label{sum_rate_lin_DL}
R_\Sigma^{\mathrm{DL,lin}}=\frac{1}{S}\sum_{\nu=0}^{S-1}\sum_{k=1}^K \log_2\big(1+\mathrm{SINR}_k^{\mathrm{DL,lin}}[\nu]\big).
\end{equation}
\begin{figure*}
\begin{align} \nonumber
  \mathrm{SINR}_k^{\mathrm{DL,lin}}[\nu]=&\frac{\kappa|\bar{\vect{h}}_k[\nu]^{\Htran} \vect{p}_k[\nu]|^2}{\kappa \sum_{i\neq k}|\bar{\vect{h}}_k[\nu]^{\Htran}\vect{p}_i[\nu]|^2+(1-\kappa)\left(\sum_{i= 1}^{K}|\bar{\vect{h}}_k[\nu]^{\Htran}\vect{p}_i[\nu]|^2+\sigma^2\right)+\kappa\sigma^2} \label{SINR_DL_lin}
  \\&= \frac{\kappa|\bar{\vect{h}}_k[\nu]^{\Htran}\vect{p}_k[\nu]|^2}{\sum_{i\neq k}|\bar{\vect{h}}_k[\nu]^{\Htran}\vect{p}_i[\nu]|^2+(1-\kappa)|\bar{\vect{h}}_k[\nu]^{\Htran}\vect{p}_k[\nu]|^2+\sigma^2} 
\end{align}
\hrulefill
\end{figure*}

\subsection{ Non-linear precoding with dirty paper coding}

The BS can use the DPC method to pre-cancel parts of the inter-user interference for the UEs, similarly to how SIC operates in the uplink. We consider DPC because it achieves the sum capacity of the broadcast channel under a total power constraint, provides a theoretical optimal nonlinear benchmark for the downlink, and enables a fundamental comparison with linear precoding in the presence of hardware distortion. 
It is established in~\cite{duality2003} that the downlink capacity region with a total power constraint is equal to the capacity region of the dual uplink multiple-access channel with a sum power constraint. This classical result holds for scenarios without hardware impairments, but we will show that the same duality principle applies in our system with hardware distortion.

The uplink sum rate with arbitrary transmit powers was stated in \eqref{eq:uplink-sum-rate-SIC}, and the power allocation must normally be performed under per-user power constraints.
The uplink sum capacity under a hypothetical  total uplink power constraint is
\begin{align}
\max_{\{\vect{D}_\rho[\nu]\}} & \frac{1}{S} \sum_{\nu=0}^{S-1} \Bigg[ \log_2 \det \left( \vect{I}_M + \frac{1}{\sigma^2} \bar{\vect{H}}[\nu] \vect{D}_\rho[\nu] \bar{\vect{H}}^{\Htran}[\nu] \right) \nonumber \\
& \quad - \log_2 \det \left( \vect{I}_M + \frac{(1-\kappa)}{\sigma^2} \bar{\vect{H}}[\nu]\vect{D}_\rho[\nu] \bar{\vect{H}}^{\Htran}[\nu] \right) \Bigg] \label{eq:ul_optimization} \\
\text{subject to} \quad & \vect{D}_\rho[\nu] \succeq 0, \quad \vect{D}_\rho[\nu] \text{ diagonal}, \ \forall \nu, \nonumber \\
& \sum_{\nu=0}^{S-1} \tr(\vect{D}_\rho[\nu]) \leq P_{\mathrm{tot}}. \label{eq:powerconstrating}
\end{align}
Here, $P_{\mathrm{tot}}$ denotes the total transmit power summed over all users and subcarriers, which is equivalent to the total downlink power constraint $\sum_{i=1}^K\sum_{\nu=0}^{S-1}\vect{\rho}_i[\nu] \leq P_{\mathrm{tot}}$ under the duality.
\begin{Lemma} \label{lemma:duality}
    The downlink rate with DPC equals \eqref{eq:ul_optimization}.
\end{Lemma}

\begin{IEEEproof}
This result is proved following the same approach as in~\cite[Sec.~III]{duality2003}, but taking distortion noise into account. The first step is to formulate the individual uplink and downlink achievable rates with SIC and DPC, respectively. Next, one can establish an uplink-downlink duality that shows that any rate achievable in the uplink can also be achieved in the downlink if the receive combining vectors are used as precoding vectors, and the same total transmit power is used. Hence, the maximum downlink sum rate equals the maximum uplink sum rate under a total power constraint.
By utilizing the optimal receive combining, one finally arrives at the expression in 
\eqref{eq:ul_optimization} under the power constraint stated in \eqref{eq:powerconstrating}.
\end{IEEEproof}

The following result shows that the downlink capacity is only limited by hardware impairments when the transmit power grows large, and not by the antenna locations.

\begin{corollary}
Under optimal power allocation and if $\bar{\mathbf{H}}[\nu]$ has rank $K$, the downlink DPC sum rate has the following upper bound in the high-SNR regime:
\begin{align}
\lim_{P_{\mathrm{tot}} \to \infty}
R_\Sigma^{\mathrm{DL,DPC}}= K \log_2 \left(\frac{1}{\mathrm{EVM}^2}\right).
\end{align}
\end{corollary}
\begin{IEEEproof}
The result follows directly by combining the uplink–downlink duality in Lemma~\ref{lemma:duality} with the high-SNR limit of the dual uplink problem, which is obtained by Proposition~\ref{cor1} for any power allocation.
\end{IEEEproof}

\section{Problem Formulation and Solution}
\label{sec:problem-formulation}
In Sections \ref{sec:uplink} and \ref{sec:uplink}, we derived achievable sum rates for the uplink and downlink with both linear and non-linear processing.
In this section, we leverage these expressions to optimize the antenna locations at the BS. Since the BS is equipped with movable antennas, the antenna positions 
$\vect{P} = [ \vect{p}_1, \ldots,  \vect{p}_M]$ can be optimized to maximize the achievable sum rate $R_{\Sigma}$. We assume antenna $m$ can be moved within a specific region $\mathcal{C}_m \subset \mathbb{R}^3$ as in the prior works~\cite{Zhu2024a,hong2025fluidantennaempowering5g} and formulate our optimization problem as follows:
\begin{align} \label{optproblem}
\underset{\vect{p}_1,\ldots,\vect{p}_M}{\text{maximize}} \quad & R_{\Sigma}(\vect{p}_1,...,\vect{p}_M), \\ \label{cona}
\text{subject to} \quad & {\vect{p}_m} \in \mathcal{C}_m, \quad 1\leq m \leq M, \\ \label{conb}
& {\| \vect{p}_m-\vect{p}_j \|}_2 \geq \lambda/2, \quad 1\leq m \neq j \leq M,
\end{align}
where (15) sets antennas $\geq$$\lambda/2$ apart to avoid mutual coupling.

We consider a sparse planar array where each antenna can be moved in non-overlapping 2D squares in the $yz$-plane, as illustrated in Fig.~\ref{simulation_setup}(a). Each square has side length $L$ and the center points are denoted as $(0,y_{m}^0,z_{m}^0)$, for $m=1,\ldots,M$. 
Hence, the set of possible positions for antenna $m$ is
\vspace{-0.5mm}
\begin{equation}
\mathcal{C}_m = \left\{ (0,y, z) \, \Big|\, |y-y_{m}^0|  \leq\frac{L}{2}, \,  | z-z_{m}^0| \leq \frac{L}{2} \right\}.
\end{equation}

\subsection{Proposed Solution Algorithm}

Next, we present our proposed algorithm for optimizing antenna positions to solve the problem formulated in \eqref{optproblem}. We utilized the particle swarm optimization (PSO) method~\cite{PSObookClerc} implemented using the Yarpiz toolbox~\cite{kalami2020ypea}, similar to our previous work~\cite{PIA2025} that considered narrowband systems. In essence, PSO considers $N_{\textrm{pt}}$ candidate solutions, known as particles, where each particle represents a unique configuration of the antenna positions. The search space is defined by the constraints in \eqref{cona} and \eqref{conb}. 
Each particle updates its position iteratively based on both its individual best-known position and the swarm's global best-known position, as outlined in Algorithm 1. The sum rate increases monotonically in each iteration, and the global best-known antenna positions are selected when the algorithm terminates ($100$ iterations were deemed sufficient in our numerical study).

\begin{algorithm}[t!]
\caption{PSO for Antenna location optimization}
\begin{algorithmic}[1]
\State Initialize particles randomly satisfying \eqref{cona}--\eqref{conb}.
\State Evaluate the sum rate for each particle.
\State Update individual and global best positions.
\State Update particle velocities and positions according to~\cite[Eq.~(18)--(19)]{PIA2025}.
\If{maximum iterations reached}
    \State Terminate.
\Else
    \State Go to Step 2.
\EndIf
\end{algorithmic}
\end{algorithm}

\section{Performance Evaluation}\label{sec:results}
In this section, the performance gains of movable antennas compared to conventional fixed antenna arrays of different shapes is evaluated from multiple perspectives. Our goal is to identify the operating conditions and propagation scenarios in which antenna mobility provides the largest benefits. We consider a BS array with $M$ movable antennas, where each antenna can move independently within a non-overlapping square region of size $5\lambda \times 5\lambda$. These square regions are arranged adjacent to each other to form a planar array structure, as illustrated in Fig.~\ref{simulation_setup}(a). 
\begin{figure*}[t!]
\centering
\begin{overpic}
[width=\columnwidth,tics=10]{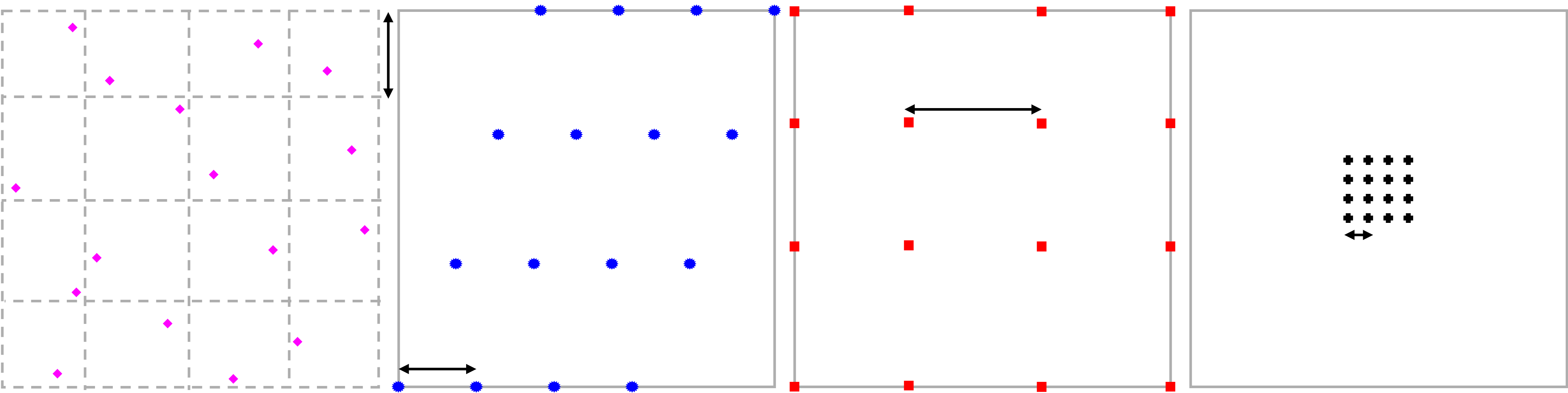} 
	 \put(32,-1.2){\footnotesize Staggered URA}
     \put(26.5,3){\footnotesize $\frac{16\lambda}{3}$}
     \put(6,-1.2){\footnotesize Movable antenna}
     \put(22.5,21.5){\footnotesize $5\lambda$}
     \put(58,-1.2){\footnotesize Sparse UPA}
     \put(60.5,19){\footnotesize $\frac{20\lambda}{3}$}
     \put(84,-1.2){\footnotesize Compact UPA}
     \put(86,8){\footnotesize $\frac{\lambda}{2}$}
\end{overpic} \vspace{-2mm}

\caption{Illustration of all considered benchmark arrays with 16 antennas. The antenna locations for the movable antennas were computed for a single random realization of the user locations. }
\label{simulation_setup}
\end{figure*}
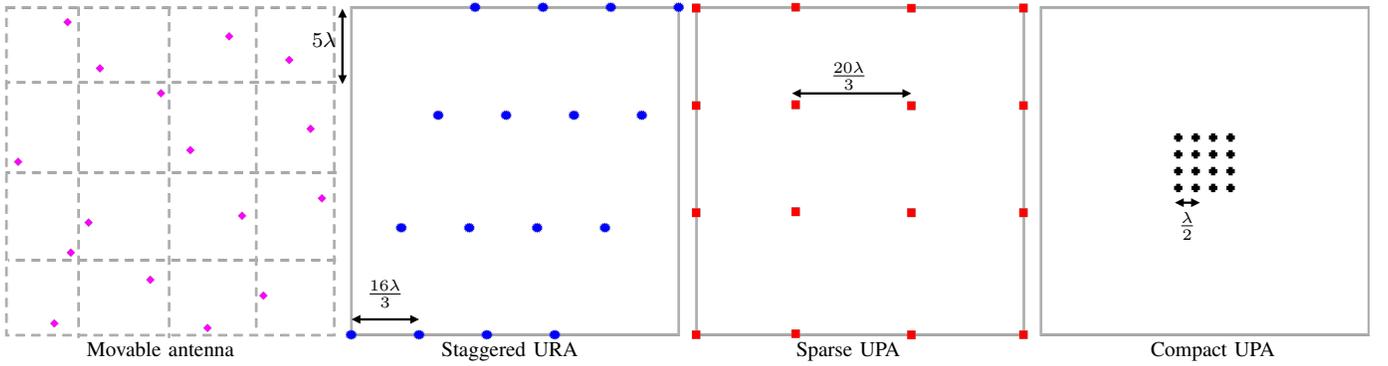
Key parameters are given in Table~\ref{Parameters}. 
\begin{table}[t!]
    \centering \vspace{-4mm}
    \caption{Summary of simulation parameters}
    \label{Parameters}
   \begin{tabular}{ p{0.55cm}||p{4.4cm}||p{2cm}}
    \hline
    \textbf{Var.} & \textbf{Description} & \textbf{Default}\\
    \hline
    $K$ & Number of users & $10$   \\
    $\sigma^2$& Noise variance (pW) & $3.98$   \\
    $M$ & Number of BS antennas & $16$\\
    $\Delta$&Subcarrier spacing (kHz)& 15\\
    $f_c$ & Carrier frequency (GHz) & 3 \\
    $N_{\textrm{pt}}$ & Number of particles in PSO& 150\\
    $\rho$ & Transmit power (mW/MHz)&1 \\
    $r$ & Radial distance for users (m) & $\sqrt{\mathcal{U}(100^2,300^2)}$  \\
    $\phi$ & Azimuth angle for users (rad) &  $\mathcal{U}(-\pi/3,\pi/3) $ \\
    \hline
    \end{tabular} \vspace{-3mm}
\end{table}

The BS is mounted on a wall with its center at a height of $4$ meters, defining the origin of the coordinate system. It serves $K=10$ outdoor users with the height of $1.25$ meters, each randomly placed in a region $(r\cos(\phi),r\sin(\phi),0)$, where the angle $\phi \sim \mathcal{U}[\phi_{\min}, \phi_{\max}]$ is uniformly distributed and the radial distance is drawn as $r \sim \sqrt{\mathcal{U}[r_{\min}^2, r_{\max}^2]}$ to ensure a uniform spatial user density in the coverage area. This setup is depicted in Fig.~\ref{user_setup} for a set of random user locations.

\begin{figure}[b!]
        \centering         
        {\includegraphics[width=1\textwidth]{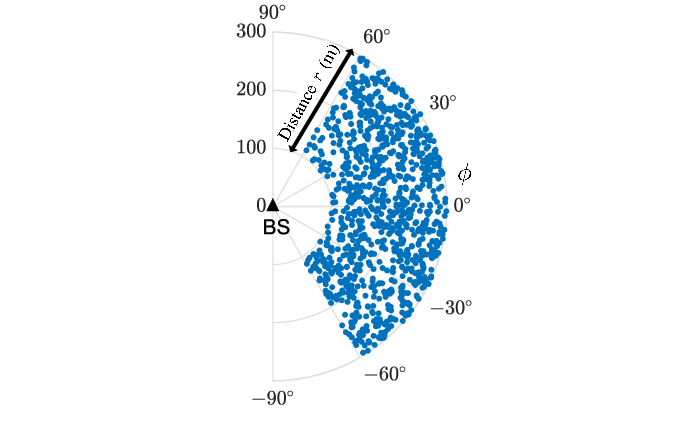}}
        \caption{Example of user distribution with 100 Monte Carlo realizations.}
\label{user_setup}  \vspace{-2mm}
\end{figure}
The following benchmark schemes are considered:

\begin{enumerate}
    \item \textbf{Zero interference}:
    A theoretical upper bound in which inter-user interference is artificially eliminated by assuming ideal spatial separation among users. This array is obtained by optimizing movable antenna positions to enforce zero mutual interference.
    
    \item \textbf{Sparse UPA}: A uniform planar array with fixed antenna positions and a large inter-element spacing of $20\lambda/3$ in both the horizontal and vertical directions. The array occupies the same overall aperture as the movable antenna system. 
    
    \item \textbf{Staggered URA}: It is a uniform rectangular array with systematically offset rows that maintains the same aperture size as the reference UPA. The horizontal coordinates of the array are equally spaced, and would coincide with those of a sparse ULA with an inter-element spacing of $20\lambda/15$. This geometry provides more uniform 2D spatial sampling than a sparse UPA and serves as a meaningful benchmark, since a movable antenna array can converge to this configuration.
    \item \textbf{Compact UPA}: It is the UPA which is placed around the origin with the traditional compact antenna spacing of $\lambda/2$.
\end{enumerate}
We consider two propagation environments: LoS-dominant and rich scattering. In the LoS-dominant case, we consider the urban microcell model in~\cite[Sec.~5.3.2]{3GPP25996}, which is used to calculate path losses and generate propagation paths in the far field. The environment includes one LoS path and six scattering clusters, each consisting of $20$ discrete paths.\footnote{The clusters are distributed uniformly in a $\pm 40^\circ$ azimuth interval and $\pm 20^\circ$ elevation interval around the LOS angles. The paths are spread in a $\pm 5^\circ$ interval around the cluster center. The delays are up to $10$ times longer than for the LoS path and their strengths are computed as in~\cite[Sec.~5.3.2]{3GPP25996}.} The Rician $\kappa$-factor is $10$\,dB and $\EVM=0.02$. In the rich scattering case, we consider an environment with $100$ small dual-path clusters uniformly distributed over all angles. The $\kappa$-factor is set to $0$\,dB, so the LoS path is only slightly stronger than the scattered paths.

\subsection{Narrowband vs. wideband}
The formulas and algorithms developed in this paper build on a wideband signal model, where the delay spread is comparable to or larger than the symbol duration, resulting in a frequency-selective (multi-path) channel. 
However, when the delay spread is much smaller than the symbol duration, the channel can be considered as frequency flat, also known as a narrowband channel~\cite{Goldsmith2005}. In this case, the multipath channel in \eqref{eq:hk-multitap} simplifies to
\begin{equation} \label{eq:narrowband}
\vect{h}_i =  \sum_{n=1}^{\Ncl} \alpha_{i,n} e^{- \imagunit 2 \pi \lambda (\tau_{i,n} - \eta)/c}  \vect{a}_{\vect{P}}(\varphi_{i,n},\theta_{i,n}),
\end{equation}
where $\vect{h}_i \in \mathbb{C}^M$ represents the narrowband channel vector for the $i$-th user.
The received signal across the array is thus a weighted superposition of array response vectors, each corresponding to one propagation path. In the narrowband case, all subcarriers experience the same channel response, and the OFDM model reduces to a single equivalent flat-fading channel per user.

To study how the transition from narrowband to wideband channels affects the average sum rate when using different array configurations, we vary the number of OFDM subcarriers, while keeping the subcarrier spacing fixed at $\Delta=15$kHz. This expands the total system bandwidth, improves delay resolution, and allows the system to resolve more of the channel's inherent frequency selectivity, effectively transitioning from a narrowband to a wideband regime. 
\begin{figure*}[t!]
\centering
\subfloat[LoS-dominant scenario.] 
{\includegraphics[width=0.49\textwidth]{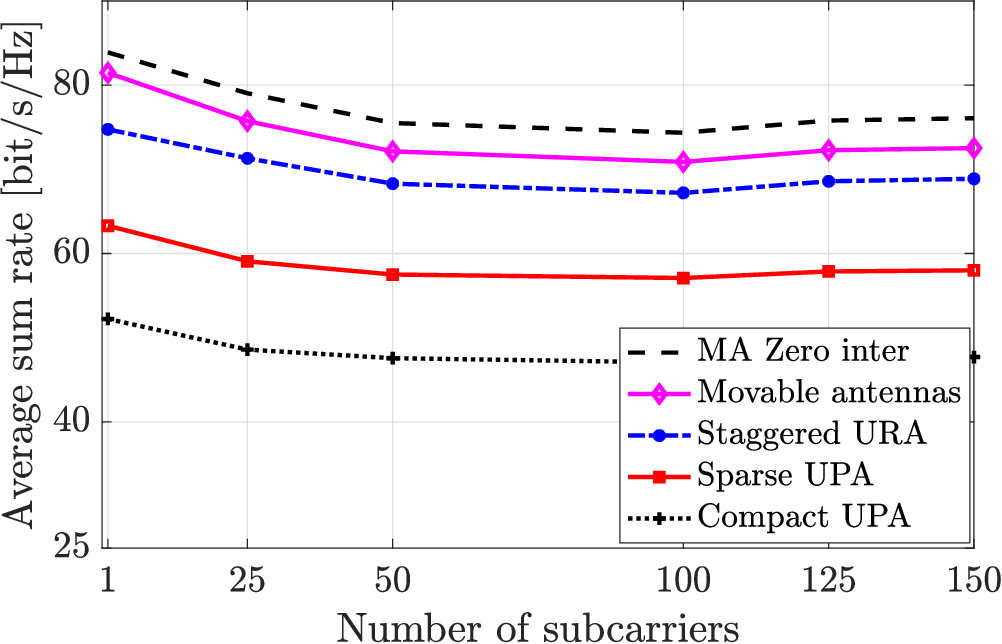}}\label{LoS_mean}\hfill
\subfloat[Rich scattering scenario.]
{\includegraphics[width=0.49\textwidth]{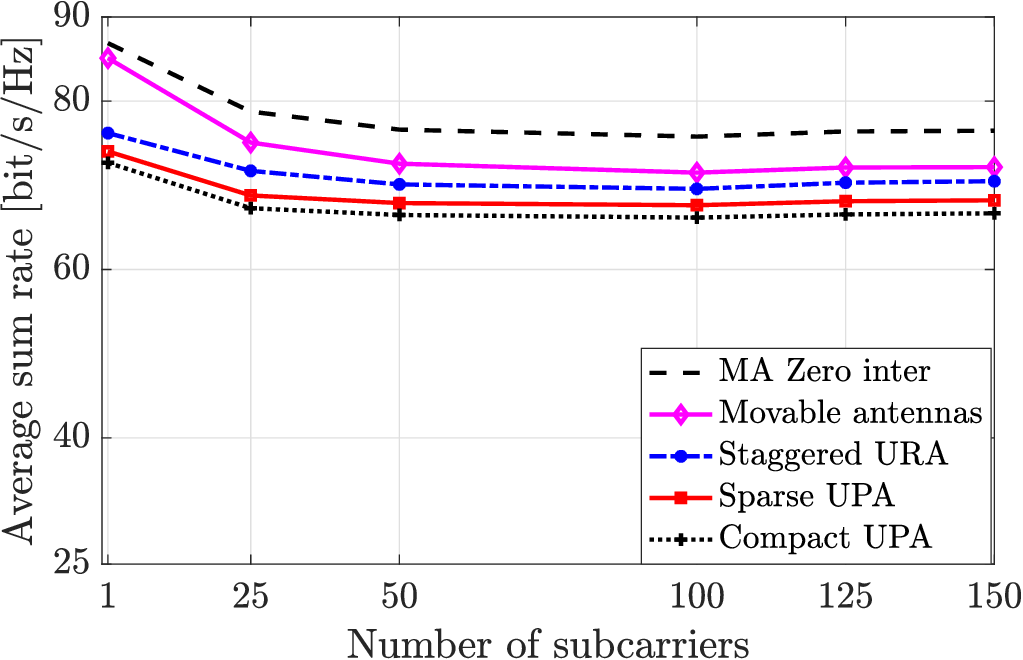}\label{NLoS_mean}}\hfill
\centering
\caption{The average uplink sum rate with SIC for a varying number of subcarriers. The BS has 16 antennas and communicates with 10 users. }
\label{mean_sumrate}
\end{figure*}

The average uplink sum rate achieved with SIC for a varying number of subcarriers is depicted in Fig.~\ref{mean_sumrate}. The results confirm that movable antennas (MAs), optimized to maximize the sum rate, as defined in \eqref{eq:sumrate-expression}, achieve the highest performance compared to all fixed benchmark schemes in both LoS-dominant and rich scattering scenarios. In Fig.~\ref{mean_sumrate}(a), corresponding to the LoS case, the MA scheme exhibits only a $2.9\%$ performance gap to the zero-interference benchmark when having a single subcarrier, whereas the compact UPA only achieves approximately $60\%$ of the upper bound with zero-interference. This performance gap persists as the number of subcarriers increases. The staggered URA is consistently the second-best array configuration, after MAs, in both LoS-dominant and rich scattering scenarios. This behavior can be explained by its staggered geometry, which preserves the overall array aperture length while also giving the antennas unique horizontal locations so that one obtains a sparse ULA if the antennas were projected onto the horizontal line. As a result, the staggered URA enables improved horizontal angular resolution and reduced spatial correlation compared to a compact UPA, enabling better separation of user channels and mitigation of inter-user interference, particularly when users are uniformly distributed within the coverage area on the ground, where horizontal resolution is more important than vertical resolution.

Fig.~\ref{mean_sumrate}(b) shows the corresponding results in the rich scattering scenario. The fixed benchmark schemes achieve a higher average sum rate than in the LoS case due to richer scattering, which creates greater channel variability, even between closely spaced users. For ease of comparison, a normalized path-loss model is used in the rich scattering simulations so the upper bound is roughly the same as in Fig.~\ref{mean_sumrate}(a).\footnote{If we were to use a different path-loss model in the rich scattering scenarios, the SNRs would generally be reduced. This results in all curves shifting downward, but would not affect the observed trends between the curves.} Notably, the performance gain of movable antennas over the fixed benchmarks decreases as the number of subcarriers increases, highlighting the effectiveness of movable antennas primarily in narrowband systems, where the number of subcarriers is less than $50$. This is attributed to the fact that MAs can create nearly orthogonal user channels on a particular subcarrier, but not simultaneously across many different ones.
 
\begin{figure*}[t!]
\centering
\subfloat[LoS-dominant scenario] 
{\includegraphics[width=0.49\textwidth]{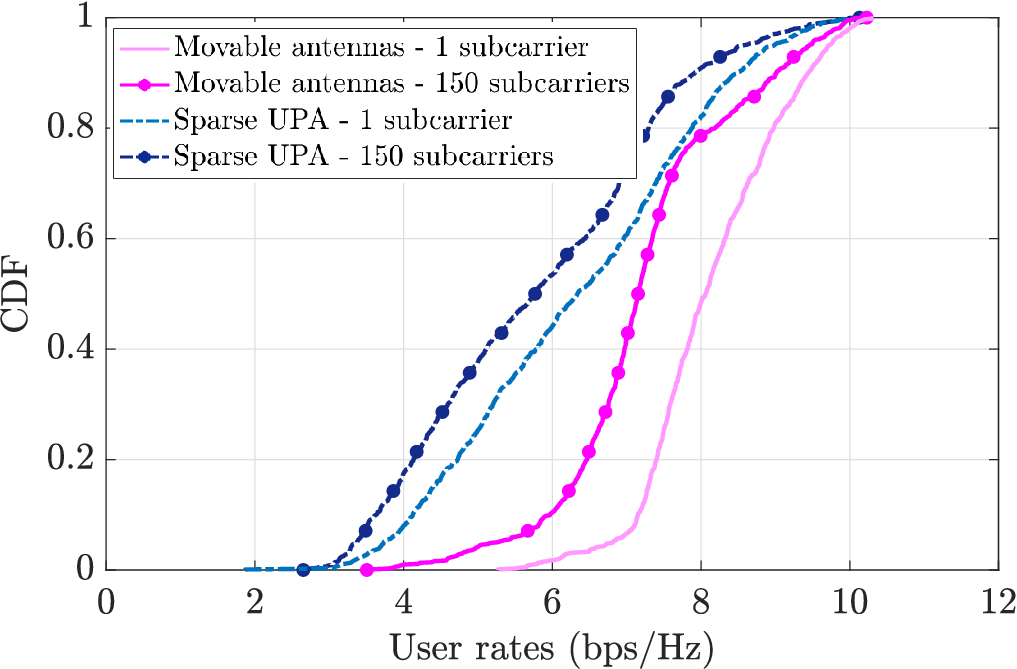}}\hfill
\subfloat[Rich scattering scenario]
{\includegraphics[width=0.49\textwidth]{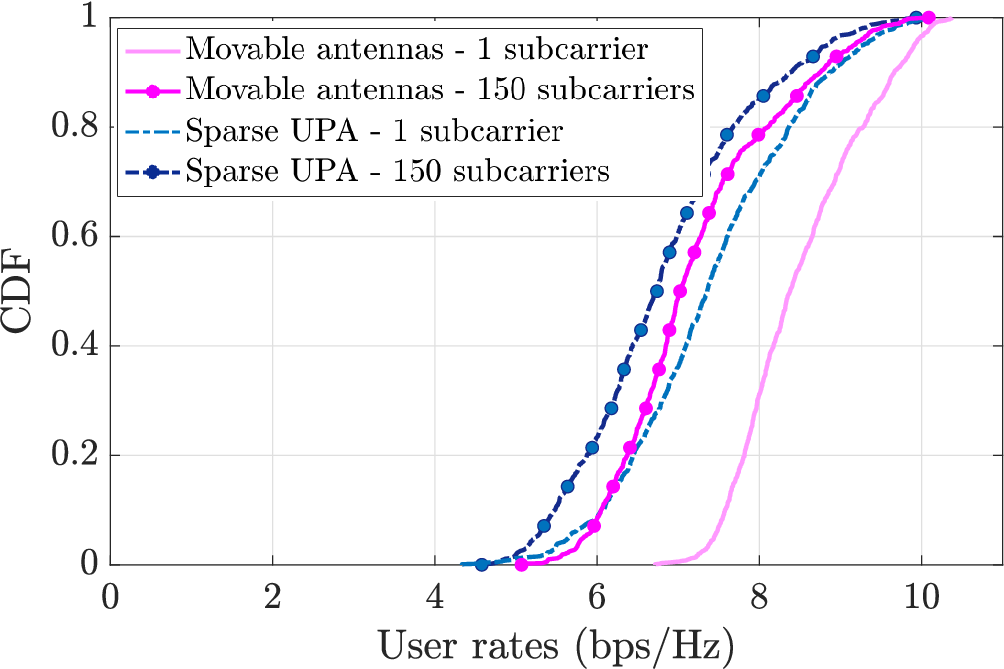}}\hfill
\centering
\caption{Variation of the user rates for two specific numbers of subcarriers (1 and 150) for both LoS-dominant and rich scattering environments, for 0.02 EVM and 10 users.}
\label{userrates}
\end{figure*}

The behaviors observed in Fig.~\ref{mean_sumrate} can be further understood by examining how the average uplink sum rate is computed, either by averaging the sum rate across all subcarriers or by aggregating the per-user rates. The subcarrier-averaged sum rate exhibits a nearly constant relative performance gap across different numbers of subcarriers in both LoS-dominant and rich scattering scenarios, indicating limited dependence on the specific number of subcarriers. For brevity, these results are not plotted in the paper.
However, when analyzing the per-user rates that contribute to the average sum rate, distinct behaviors emerge. These results are shown in Fig.~\ref{userrates}, which presents the cumulative distribution function (CDF) of per-user rates for random user locations.
 
In the rich scattering scenario, the variability of the per-user rates remains relatively stable as the number of subcarriers increases. This behavior is due to the rich scattering environment, which results in low channel correlation in the frequency domain and leads to sufficiently diverse channel realizations across different subcarriers, even when only a few subcarriers are considered. In contrast, in the LoS-dominant scenario, the per-user rates depend more strongly on the number of subcarriers because the dominant propagation paths induce strong frequency correlation across subcarriers. As a result, the effective channel gains become more sensitive to the frequency sampling resolution, leading to larger disparities between users. This indicates that there are some combinations of user locations that are challenging to handle even with movable antennas---unless one gets richer multipaths as in the rich scattering case.

The primary objective of this section is to identify the most relevant practical use case for movable antennas. Therefore, the remainder of this section focuses on the LoS scenario, where the array geometry has the greatest impact. To balance accuracy and computational efficiency, $50$ subcarriers are used in subsequent simulations, as the average sum-rate performance exhibits a consistent trend as the subcarrier number is increased beyond this point.

\subsection{Impact of hardware impairments}

Next, we investigate the impact of hardware impairments on the considered array configurations. Fig.~\ref{figure_simulationEVM} shows the average uplink sum rate with SIC for different EVM values. As the EVM increases, the sum rate of all schemes decreases and converges to the theoretical upper bound derived in \eqref{eq:upper-limit_evm}. This behavior confirms that hardware impairments uniformly degrade system performance and impose a common rate ceiling, regardless of the array configuration.
Moreover, while movable antennas provide noticeable performance gains in situations with low impairments (smaller EVM), these gains diminish as the system becomes hardware-limited (larger EVM). Consequently, movable antennas are most beneficial when the hardware impairments are sufficiently small, whereas their advantage over fixed arrays vanishes in the high impairment regime.

\begin{figure}[t!]
        \centering         {\includegraphics[width=1\textwidth]{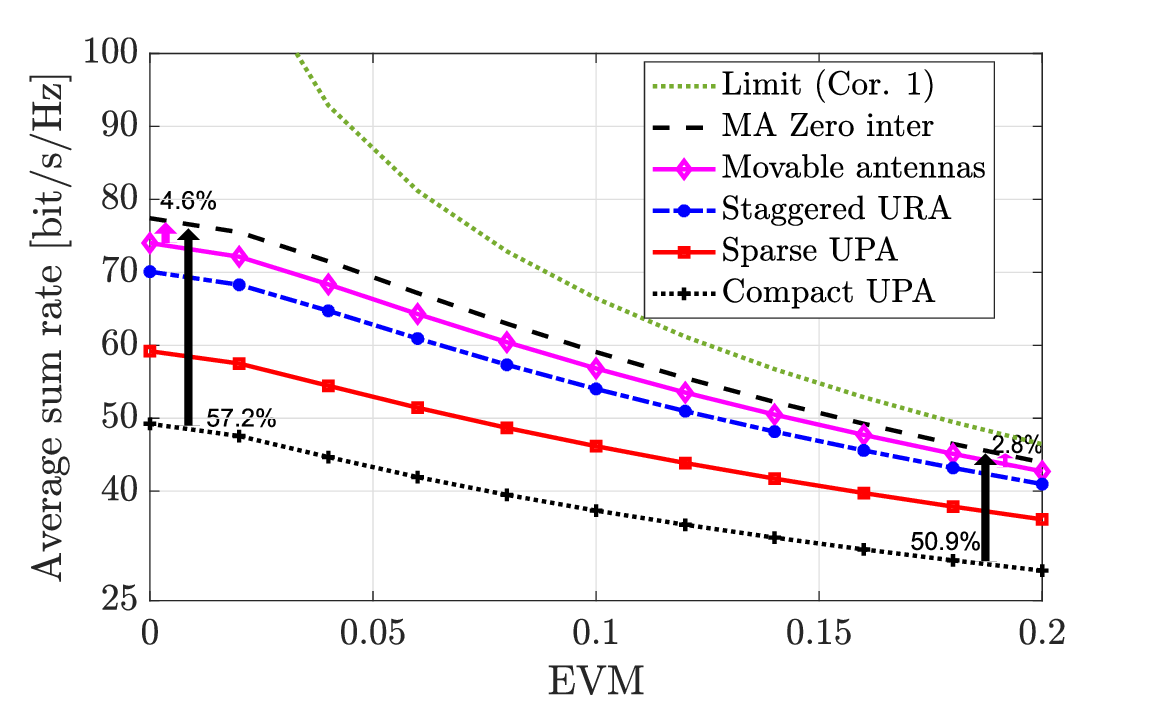}}
        \caption{Average sum rate versus the EVM.}
\label{figure_simulationEVM}  \vspace{-2mm}
\end{figure}

\subsection{Linear vs non-linear processing}

Next, we examine whether the performance gains and optimal configuration of movable antennas depends on the choice of receiver processing or not. Here, we compute uplink linear processing, achieving the sum rate defined in \eqref{sum_rate_lin}, and nonlinear processing based on SIC, achieving the sum rate in \eqref{eq:sumrate-expression}. The results are shown in Fig.~\ref{figure_linvsnonlin_uplink} for LoS propagation environment with $50$ subcarriers and an EVM of $0.02$. 

We optimize the MA configuration using one processing method and then evaluate its performance using the other. When the MA system is optimized for SIC and subsequently operates using linear processing, it achieves nearly the same sum-rate performance as when it is optimized directly for linear processing. A similar observation holds in the reverse case: optimizing the MA configuration for linear processing and evaluating it using SIC achieves performance comparable to MA configurations optimized explicitly for SIC. 

These results indicate that the MA optimization is not sensitive to the specific uplink processing scheme. This is because the MA array shapes the effective multi-user channel geometry such that the channel is well-conditioned and user separability (low inter-user correlation) is obtained, mostly irrespective of the receiver processing methods. These properties benefit both linear processing and SIC, so the configurations optimized for either method are nearly identical. The remaining performance gap between SIC and linear processing exists because SIC can further remove the residual multiuser interference at detection, whereas linear processing cannot. Consequently, employing low-complexity linear processing during MA optimization can significantly reduce computational complexity and training time without sacrificing end performance. 

As expected, all schemes achieve higher sum rates under SIC than under linear processing, as expected since non-linear processing is generally needed to achieve capacity. For example, in a system with $K=10$ users and $M=16$ antennas, when non-linear processing is employed at the receiver, a staggered URA with SIC achieves sum rate performance that is only marginally lower than that of a MA array (regardless of whether the MA geometry is optimized under linear or non-linear processing). In such scenarios, the limited performance gap suggests that a fixed antenna structure may be preferable, since the additional mechanical and hardware complexity required to enable antenna mobility may not be justified.

\begin{figure}[t!]
        \centering         
        {\includegraphics[width=1\textwidth]{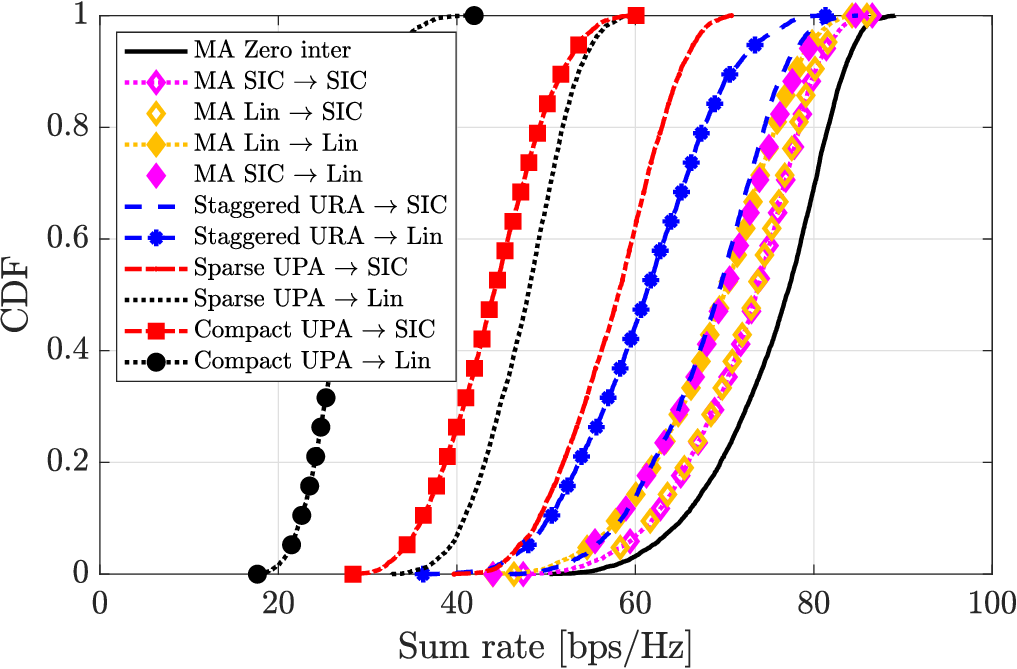}}
        \caption{Uplink linear (shown with filled markers) vs non-linear processing.}
\label{figure_linvsnonlin_uplink}  \vspace{-2mm}
\end{figure}

Further insights can be obtained by examining the average user rates for both processing methods, as shown in Fig.~\ref{figure_userrates_linvsnonlin_uplink}. Under linear processing, the average user rates are relatively uniform across users. In contrast, SIC introduces a hierarchical structure due to the decoding order: users decoded later experience less residual interference and therefore tend to achieve higher rates. Minor fluctuations observed in the curves are due to the finite number of Monte Carlo realizations and random channel variations.

\begin{figure}[t!]
        \centering 
{\includegraphics[width=1\textwidth]{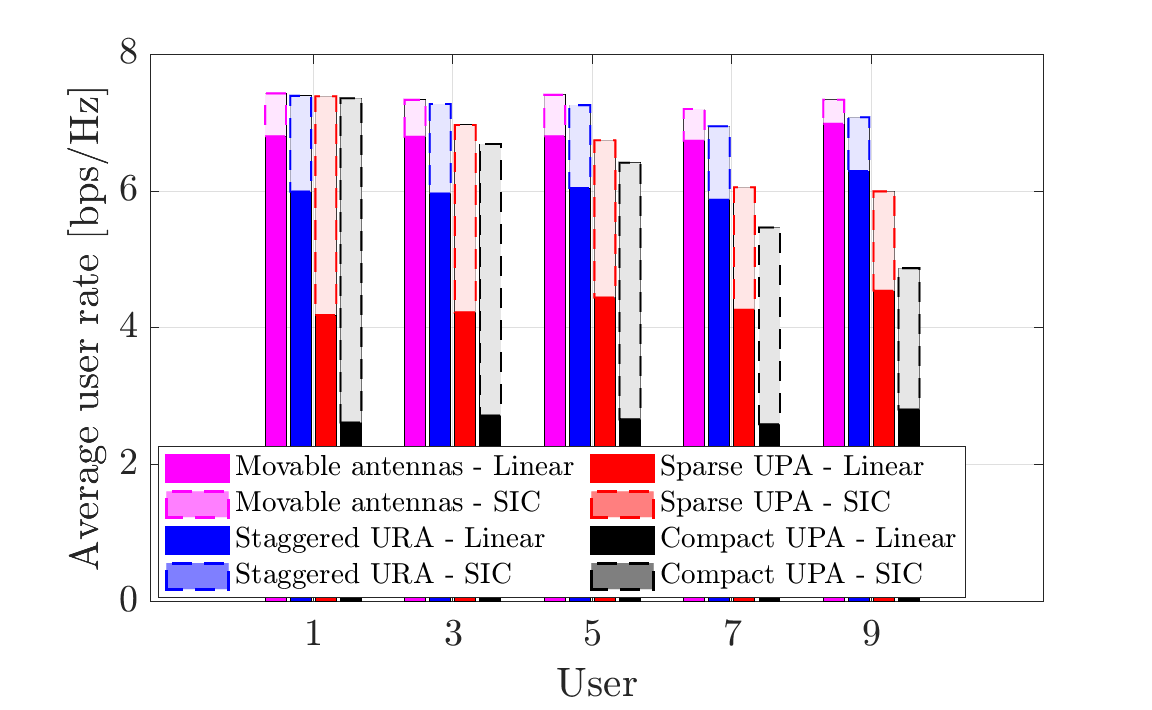}}
        \caption{Average user rates comparison between uplink linear vs non-linear processing for odd number of users.}
\label{figure_userrates_linvsnonlin_uplink}  \vspace{-2mm}
\end{figure}

While the above results focus on a fixed number of users, it is also important to examine how the relative gains of movable antennas evolve with the system load. To assess this impact of the number of users, Fig.~\ref{figure_sumrate_vs_sumrate} shows the uplink sum rate under linear processing as a function of the number of users. When the number of users is small, all considered antenna configurations achieve nearly identical performance, indicating that movable antennas provide limited additional gains in lightly loaded systems. As the number of users increases, the performance gap between movable antenna arrays and fixed antenna structures becomes progressively more pronounced. This behavior highlights that the benefits of antenna mobility manifest primarily in moderate-to-heavily loaded scenarios, where spatial interference becomes a dominant performance bottleneck. In contrast, in the regimes with only a few users deployed, the additional degrees of freedom offered by movable antennas are underutilized, so the additional mechanical and hardware complexity associated with antenna mobility may not always be justified.
\begin{figure}[t!]
        \centering 
{\includegraphics[width=1\textwidth]{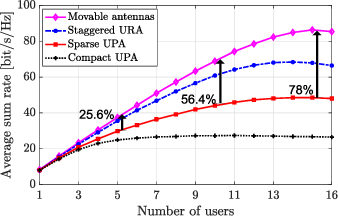}}
        \caption{Average sum rate for varying number of users (1 to 16) with 50 subcarriers, LoS-dominant case, and linear processing.}
\label{figure_sumrate_vs_sumrate}  \vspace{-2mm}
\end{figure}

\subsection{TDD vs FDD operation}
After analyzing the independence of MA optimization with respect to the receiver processing, we consider the difference between uplink and downlink communication. This can be implemented with different duplexing modes, specifically TDD and FDD. In TDD systems, the uplink and downlink share the same carrier frequency and are separated in time, whereas in FDD systems, the uplink and downlink transmissions occur simultaneously over different carrier frequencies.

\subsubsection{TDD}

We now investigate the performance of the MA system in a TDD setting for both uplink and downlink transmissions, and compare it with the considered fixed array benchmark schemes. We consider the performance achieved with linear processing for both uplink and downlink transmissions, using the sum-rate expressions in \eqref{sum_rate_lin} and \eqref{sum_rate_lin_DL}, respectively. A carrier frequency of 3~GHz is used for both links. For the downlink, a base-station transmit power spectral density of $20$~mW/MHz is assumed, while for the uplink, each user transmits with $1$~mW/MHz. The CDF of the sum rate for different realizations of user locations, LoS scenario with $50$ subcarriers and $0.02$ EVM are shown in Fig.~\ref{figure_power_scaling_dlvsul}. Similar to the receiver-processing study, we optimize the MA configuration in one transmission direction and evaluate its performance in the other. Specifically, the MA configuration is optimized for uplink transmission and then tested for both uplink and downlink operation, and vice versa. The results indicate that the MA system performance is largely independent of the transmission direction in the TDD case. This observation implies that once the antenna positions are optimized in one direction, for example in the uplink, the same configuration can be reused for downlink transmission without re-optimization. This significantly reduces computational complexity, system overhead and latency, as the system can switch between uplink and downlink faster without the delay caused by repositioning the antennas.

The figure also shows what happens when the total downlink transmit power is increased by $10$ or $100$ times, thereby creating a large power imbalance between uplink and downlink (as often encountered in practice). The results show that the MA system's performance is robust to the choice of transmission direction, even when the transmission power is increased by one or two orders of magnitude. Increasing the downlink transmit power uniformly raises the achievable sum rates of all schemes, but does not change their relative gaps and the same scheme remains best regardless of power level. This is because optimized MAs suppress inter-user interference by finding a favorable spatial configuration, this interference structure and thus the relative performance across schemes, remains largely unchanged when transmit power is scaled up.

\begin{figure}[t!]
        \centering         {\includegraphics[width=1\textwidth]{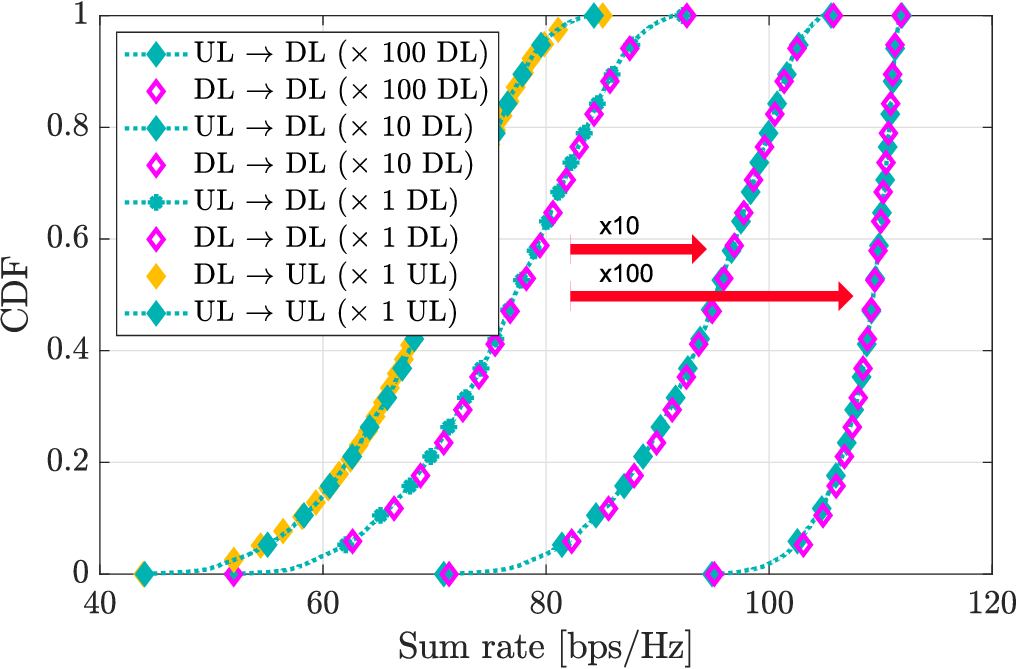}}
        \caption{CDF of the sum rate achieved by MAs in uplink or downlink transmission. The notation $X  \rightarrow Y$ indicates that the antenna configuration is optimized in $X$ direction and evaluated in $Y$ direction. DL test powers of $\times1$, $\times10$, and $\times100$ relative to 20 mW/MHz BS reference. The UL reference is 1 mW/MHz per UE.}
\label{figure_power_scaling_dlvsul}  \vspace{-2mm}
\end{figure}

\subsubsection{FDD}

The situation may differ in FDD operation since the channel vectors are frequency-dependent. Specifically, in Fig.~\ref{figure_FDD}, the downlink carrier frequency is fixed at 3~GHz, while the uplink carrier frequency is varied from 3.0 to 2.7~GHz (a 10\% relative difference). In other words, in the FDD case, the movable antenna configuration is optimized at one carrier frequency, and the resulting antenna positions are then tested for transmission at another carrier frequency without re-optimization. We observe in Fig.~\ref{figure_FDD} that as the UL–DL frequency separation increases, the MA system exhibits a gradual reduction in sum-rate performance, reflecting sensitivity to the frequency gap. Under the assumed geometry, the primary source of this degradation is the wavelength-dependent antenna spacing effect rather than changes in scattering structure. Considering frequency-dependent geometric parameters would likely further degrade performance, reducing the gap between optimized and fixed benchmark schemes. Nevertheless, even under increasing frequency separation, the MA system consistently outperforms the fixed staggered URA benchmark in both uplink and downlink transmissions. Note that there is a power imbalance between the uplink and downlink, so only curves of the same color should be compared.

\begin{figure}[t!]
        \centering         {\includegraphics[width=1\textwidth]{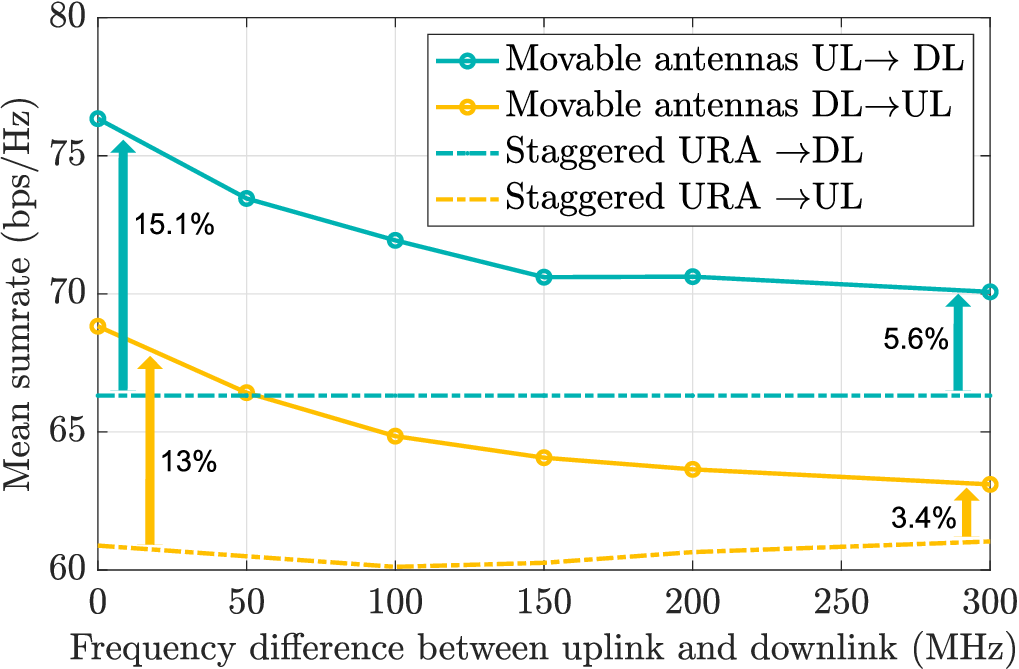}}
        \caption{Average sum rate for MA and staggered URAs for varying frequency difference between UL and DL. The notation $X  \rightarrow Y$ indicates that the antenna configuration is optimized in $X$ and evaluated in $Y$.}
\label{figure_FDD}  \vspace{-2mm}
\end{figure}

\section{Discussion and Conclusion}\label{Sect_conclusion}
This paper presented a comprehensive performance evaluation of movable antenna (MA) arrays under realistic uplink and downlink operation. For this purpose, a unified wideband system model was developed, enabling analytical characterization of the uplink and downlink sum rate under different receiver and transmitter processing schemes. The analysis specifically accounted for various configurations involving: (i)~wideband OFDM transmission, (ii)~linear and nonlinear processing, (iii)~hardware impairments, (iv)~user loading effects, and (v)~both time division duplexing (TDD) and frequency division duplexing (FDD). The key objective of this study was to identify the practical scenarios and configurations in which antenna mobility can provide \emph{substantial performance benefits} over conventional fixed MIMO arrays. The robustness of MA systems was further investigated using an EVM-based hardware impairment model, and an analytical high-SNR sum rate ceiling was derived.

The presented study led to the following key conclusions regarding MA suitability and achievable gains in practical scenarios. \emph{First}, it has been confirmed that even in the wideband setup, the MA arrays configured for sum-rate maximization can offer performance gains over fixed benchmark configurations in certain favorable propagation conditions, particularly in LoS-dominant environments with strong multi-user interference. \emph{Second}, our analysis revealed that hardware impairments impose a common performance limit across all array configurations, and that the gains offered by antenna mobility diminish as the system becomes hardware-limited.

\emph{Third}, the interaction between MA optimization and receiver processing was also examined. The results demonstrated that antenna position optimization is largely independent of the specific uplink processing method, suggesting that low-complexity linear processing can be employed during MA optimization without significantly affecting the final performance. \emph{Fourth}, the impact of user loading was analyzed, showing that the benefits of antenna mobility become more pronounced as the number of users increases and spatial interference becomes the dominant performance bottleneck. 

\emph{Finally}, MA systems were evaluated under duplexing constraints. In TDD operation, channel reciprocity enables strong uplink--downlink consistency, allowing antenna configurations optimized in one transmission direction to be reused in the other without re-optimization. In contrast, FDD operation appears to exhibit sensitivity to uplink--downlink frequency separation due to channel frequency dependence.

Overall, our results indicate that movable antennas can deliver performance gains in practical wideband multi-user systems in specific operating conditions, but their effectiveness depends strongly on propagation characteristics, hardware quality, and system complexity. 
The performance gains over well-designed fixed arrays are small in rich-scattering environments, when the user load is low, or when the hardware quality is poor. These are examples of practical situations where the extra costs of deploying and operating MAs might not outweigh the performance benefits.
These findings provide valuable insights into when antenna mobility is likely to be beneficial and when conventional fixed array configurations are a more reasonable choice. An interesting direction for future work is to extend the proposed framework to account for channel estimation errors and user mobility, enabling the evaluation of MA performance in dynamic scenarios.

\balance
\bibliographystyle{IEEEtran}
\bibliography{IEEEabrv,mybib}

\end{document}